\documentclass[11pt]{article}

\usepackage[latin1]{inputenc}    
\usepackage[OT1]{fontenc}
\usepackage[english]{babel}
\usepackage{fullpage}
\usepackage{amsmath}
\usepackage{amssymb}
\usepackage{amsthm}
\usepackage{mathrsfs}
\usepackage{graphicx}
\usepackage{dsfont}
\usepackage{wrapfig}
\usepackage{comment}
\usepackage{rotating}
\usepackage{hyperref} 
\usepackage{subfigure}
\usepackage{algorithmic}
\usepackage[ruled,linesnumbered,shortend,vlined,noend]{algorithm2e}
\usepackage[all]{xy}

\newcommand{\R}{\mathbb{R}}

\newcommand{\Rext}{\mathbb{R}_{\text{ext}}}
\newcommand{\N}{\mathbb{N}}

\newcommand{\Dg}{\mathrm{Dg}}

\newcommand{\kPSS}{k_{\rm PSS}}
\newcommand{\kPWG}{k_{\rm PWG}}
\newcommand{\kSW}{k_{\rm SW}}
\newcommand{\SW}{{\rm SW}}

\newcommand{\SpD}{{\mathcal D}}

\newcommand{\SpfbD}{{\mathcal D}_{\rm f}^{\rm b}}
\newcommand{\SpND}{{\mathcal D}_N^{\rm b}}

\newtheorem{thm}{Theorem}[section]

\newtheorem{prop}[thm]{Proposition}
\newtheorem{lem}[thm]{Lemma}
\newtheorem{defin}[thm]{Definition}
\newtheorem{rem}[thm]{Remark}

\newcommand{\eqdef}{\ensuremath{\stackrel{\mbox{\upshape\tiny def.}}{=}}}

\title{Sliced Wasserstein Kernel for Persistence Diagrams}
\author{Mathieu Carri\`ere, Marco Cuturi, Steve Oudot}
\date{}

\begin{document}

\maketitle

\begin{abstract}
Persistence diagrams play a key role in topological data
analysis (TDA), in which they are routinely used to describe
 topological properties of complicated shapes. persistence diagrams
enjoy strong stability properties and have proven their utility in
various learning contexts.  They do not, however, live in a space
naturally endowed with a Hilbert structure and are usually compared
with non-Hilbertian distances, such as the bottleneck distance. To
incorporate persistence diagrams in a convex learning pipeline, several kernels have been
proposed with a strong emphasis on the stability of the resulting RKHS
distance w.r.t. perturbations of the persistence diagrams.  In this article, we use the
Sliced Wasserstein approximation of the Wasserstein distance to
define a new kernel for persistence diagrams, which is not only provably stable but
also discriminative (with a bound depending on the number of points in the persistence diagrams) 
w.r.t. the first diagram distance between persistence diagrams. 
We also demonstrate its practicality, by
developing an approximation technique to reduce kernel computation
time, and show that our proposal compares favorably to existing
kernels for persistence diagrams on several benchmarks.
\end{abstract}

\section{Introduction}

Topological Data Analysis (TDA) is an emerging trend in data
science, grounded on topological methods to design descriptors
for complex data---see e.g.~\cite{Carlsson09b} for an introduction to
the subject.  The descriptors of TDA can be used in various contexts,
in particular statistical learning and geometric inference, where they
provide useful insight into the structure of data.  Applications
of TDA can be found in a number of scientific areas, including
computer vision~\cite{Li14}, materials science~\cite{Hiraoka16}, and
brain science~\cite{Singh08}, to name a few.  The tools developed in
TDA are built upon persistent homology
theory~\cite{Edelsbrunner10,Oudot15}, and their main output is a
descriptor called {\em persistence diagram}, which encodes the
topology of a space at all scales in the form of a point cloud with
multiplicities in the plane $\R^2$---see Section~\ref{sec:persHom} for more details.

\paragraph{Persistence diagrams as features.} 
The main strength of persistence diagrams is their stability with respect to perturbations of the data~\cite{Chazal09c,Chazal13b}.
On the downside, their use in learning tasks is not straightforward.
Indeed, a large class of learning methods, such as SVM or PCA, requires
a Hilbert structure on the descriptors space, which is not the case
for the space of persistence diagrams. Actually, many simple operators of $\R^n$, such
as addition, average or scalar product, have no analogues in that
space. Mapping persistence diagrams to vectors in $\R^n$ or in some infinite-dimensional Hilbert space 
is one possible approach to facilitate their use in discriminative settings.

\paragraph{Related work.} A series of recent contributions have proposed kernels for persistence diagrams,
falling into two classes.  The first class of methods builds explicit
feature maps: one can, for instance, compute and sample functions extracted from
persistence diagrams~\cite{Bubenik15,Adams17,Robins16}; sort the entries of
the distance matrices of the persistence diagrams~\cite{Carriere15a}; treat the
points of the persistence diagrams as roots of a complex polynomial, whose coefficients are
concatenated~\cite{diFabio15}.
%
The second class of methods, which is more relevant to our work, defines implicitly feature maps by focusing instead on building kernels for persistence diagrams. For
instance, \cite{Reininghaus15} use
 solutions of the heat differential equation in the plane and compare them
with the usual $L^2(\R^2)$ dot product.
~\cite{Kusano16} handle a persistence diagram as a discrete measure on the plane,
and follow by using kernel mean embeddings with Gaussian kernels---see
Section~\ref{sec:expe} for precise definitions.
Both kernels are provably {\em stable}, in the sense that the metric they induce in their respective reproducing kernel Hilbert space (RKHS) 
is bounded above by the distance between persistence diagrams. 
Although these kernels are injective, there is no evidence that their induced RKHS distances are discriminative and therefore follow 
the geometry of the diagram distances, which are more widely accepted distances to compare persistence diagrams.

More generally, one of the reasons why the derivation of kernels for persistence diagrams is not  
straightforward is that the natural metrics between persistence diagrams, the {\em diagram distances}
are not negative semi-definite. Indeed, these diagram distances are very similar to the 
{\em Wasserstein distance}~\cite[\S6]{Villani09} between probability measures, which is not negative
semi-definite. However, a relaxation of this metric called the {\em Sliced Wasserstein distance}~\cite{Rabin11}
has recently been shown to be negative semi-definite and was used to derive kernels for probability distributions
in~\cite{Kolouri16}.    

\paragraph{Contributions.} In this article, 
we use the Sliced Wasserstein
distance of~\cite{Rabin11} to define a new kernel for persistence diagrams, which we prove
to be both stable and discriminative. Specifically, we provide distortion bounds on the Sliced Wasserstein distance that quantify its ability to 
mimic the diagram distances between persistence diagrams. This
is in contrast to other kernels for persistence diagrams, which only focus on
stability. We also propose a simple approximation algorithm
to speed up the computation of that kernel, confirm experimentally its discriminative power and show 
that it outperforms experimentally both proposals of \cite{Kusano16} and \cite{Reininghaus15} in several supervised classification problems.

\section{Background}

\subsection{Persistent Homology}
\label{sec:persHom}

\newcommand{\dg}{\operatorname{Dg}}
\newcommand{\distb}{\operatorname{d_b}}
\newcommand{\cost}{\operatorname{c}}

\begin{figure*}
\centering
\subfigure[]{
\centering
\label{fig::sublevel_sets}
 \includegraphics[height=4.2cm]{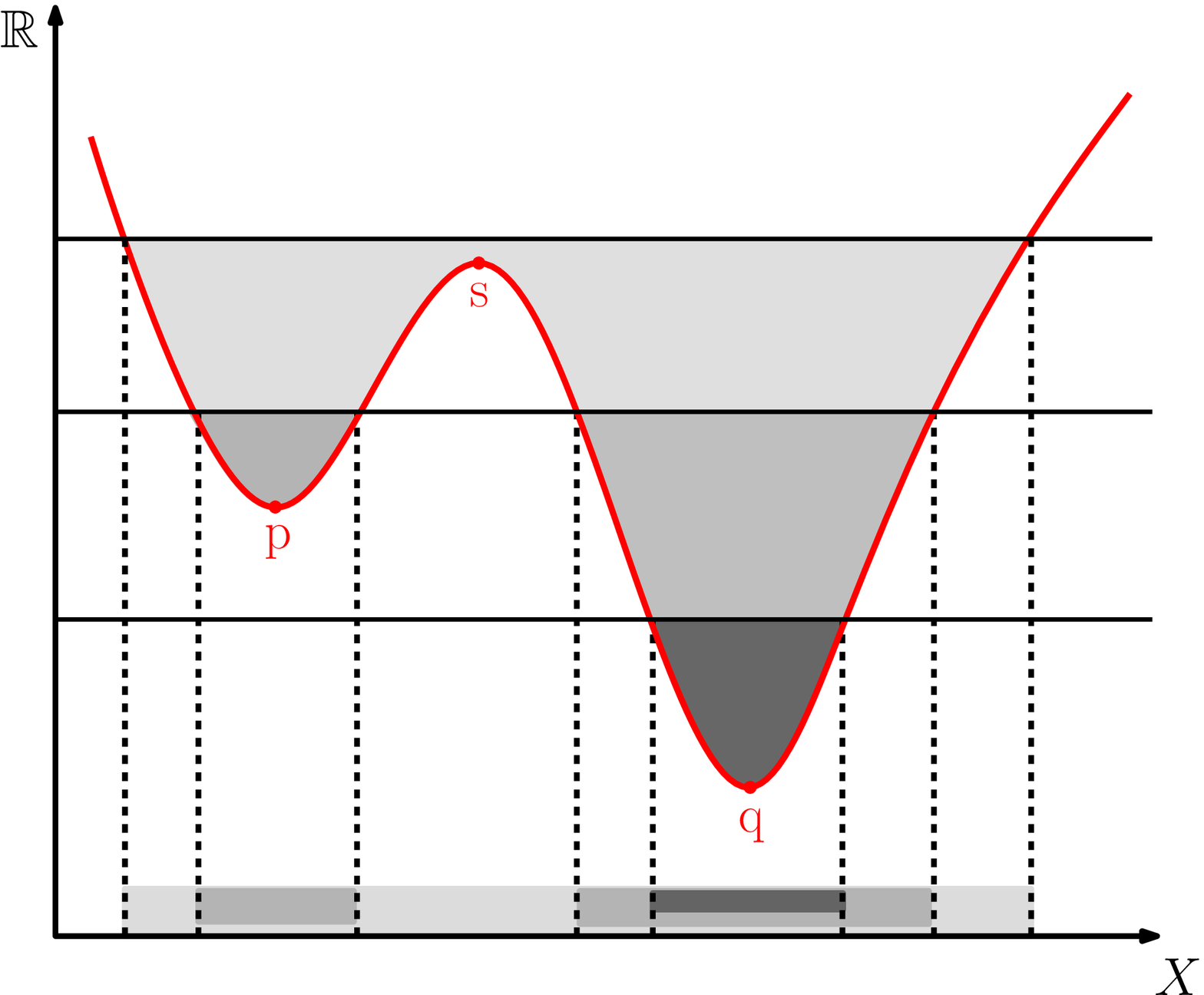}
}
\subfigure[]{\label{fig::barcode1}
\includegraphics[height=4.2cm]{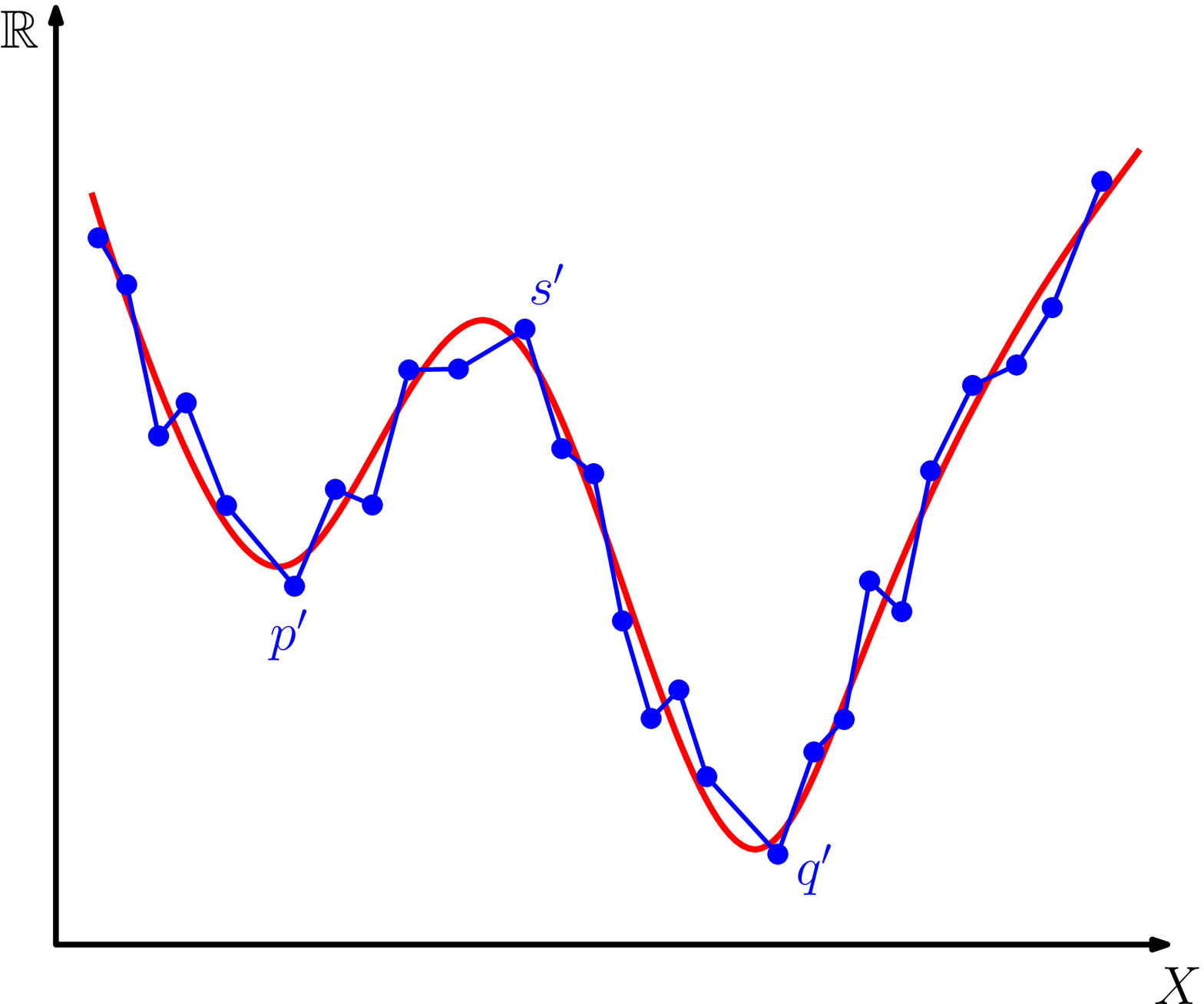}
}
\subfigure[]{\label{fig::barcode2}
\includegraphics[height=4.2cm]{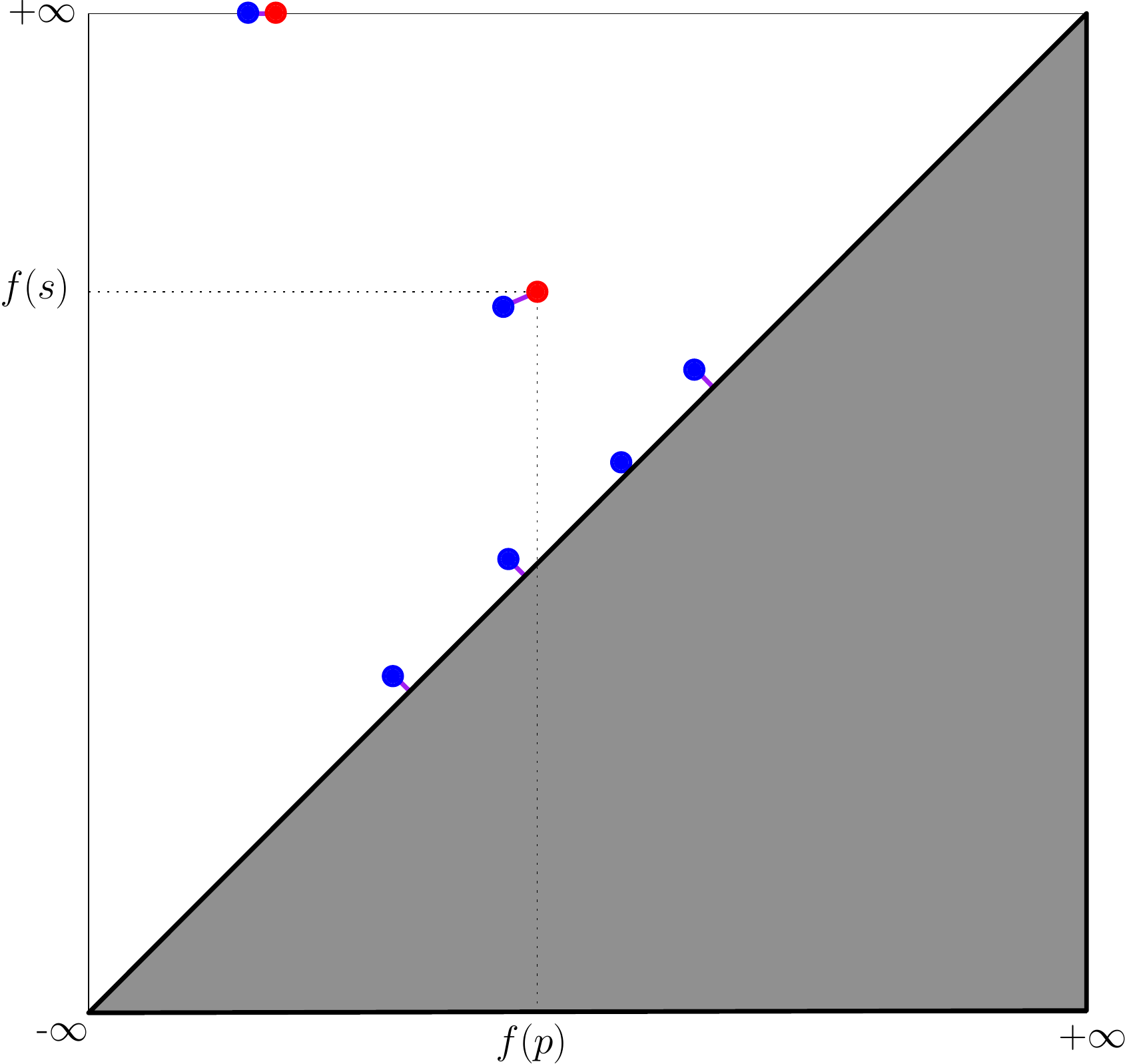}
}
\caption{\label{fig:barcode}Sketch of persistent homology: 
(a)~the horizontal lines are the boundaries of sublevel sets $f((-\infty,t])$, which are colored in decreasing shades of grey.
The vertical dotted lines are the boundaries of their different connected components.
For instance, a new connected component is created in the sublevel set $f^{-1}((-\infty,t])$ 
when $t=f(p)$, and it is merged (destroyed) when $t=f(s)$; its lifespan is represented by a copy of the point with coordinates
$(f(p), f(s))$ in the persistence diagram of~$f$ (Figure~(c)); 
(b)~a piecewise-linear approximation $g$ (blue) of the function $f$ (red) from sampled values; 
(c)~superposition of $\Dg(f)$ (red) and $\Dg(g)$ (blue), showing the partial matching of minimum cost (magenta) between the
two persistence diagrams.}
\end{figure*}

Persistent homology is a technique inherited from algebraic
topology for computing stable descriptors on real-valued
functions. Given $f:X\to\R$ as input, persistent homology
outputs a planar point set with multiplicities, called the 
{\em  persistence diagram} of~$f$ and denoted by $\Dg(f)$.
Note that the coordinates of the points belong to the extended real line $\Rext=\R\cup\{+\infty\}$. 
See Figure~\ref{fig:barcode} for an example. 
To understand the meaning of each point in this diagram, it suffices to know that, to compute $\Dg(f)$, 
persistent homology considers the family of {\em sublevel sets} of
$f$, i.e. the sets of the form $f^{-1}((-\infty, t])$ for $t\in\R$,
  and it records the {\em topological events} (e.g. creation or merge
  of a connected component, creation or filling of a loop, void, etc.)
  that occur in $f^{-1}((-\infty, t])$ as $t$ ranges from $-\infty$ to
    $+\infty$.  Then, each point $p\in \Dg(f)$ represents the lifespan
    of a particular {\em topological feature} (connected component,
    loop, void, etc.), with its creation and destruction times as
    coordinates. See again Figure~\ref{fig:barcode} for an
    illustration.

For the interested reader, we point out that the mathematical tool
used by persistent homology to track the topological events in the
family of sublevel sets is {\em homological algebra}, which turns the
parametrized family of sublevel sets into a parametrized family of
vector spaces and linear maps. Computing persistent homology then
boils down to computing a family of bases for the vector spaces, which
are compatible with the linear maps. It will be no surprise to the
reader familiar with matrix reduction techniques that the simplest way
to implement the compatible basis computation is using Gaussian elimination. 

\paragraph*{Distance between persistence diagrams.}
We now define the {\em $p$th diagram distance} between persistence diagrams.
Let $p\in\mathbb{N}$ and $\Dg_1,\Dg_2$ be two persistence diagrams. 
Let $\Gamma:\Dg_1\supseteq A\rightarrow B\subseteq \Dg_2$ be a {\em partial bijection} between $\Dg_1$ and $\Dg_2$.
Then, for any point $x\in A$, the $p$-{\em cost} of $x$ is defined as $c_p(x)=\|x-\Gamma(x)\|_\infty^p$,
and for any point $y\in(\Dg_1\sqcup \Dg_2)\setminus (A\sqcup B)$, the $p$-{\em cost} of $y$ is defined as 
$c'_p(y)=\|y-\pi_\Delta(y)\|_\infty^p$, where $\pi_\Delta$ is the projection 
onto the diagonal~$\Delta=\{(x,x) : x\in\R\}$.
The cost $\cost_p(\Gamma)$ is defined as:
$\cost_p(\Gamma)=(\sum_x c_p(x) + \sum_y c'_p(y))^{1/p}.$
We then define the {\em $p$th diagram distance}
 $d_p$ as the cost of the best partial bijection: 
$$d_p(\Dg_1,\Dg_2) = \inf_\Gamma \cost_p(\Gamma).$$
In the particular case $p=+\infty$, the cost of $\Gamma$ is defined as 
$\cost(\Gamma)=\max\{\max_x c_1(x) + \max_y c'_1(y)\}.$
The corresponding distance $d_\infty$ is often called the {\em bottleneck distance}.
One can show that $d_p\rightarrow d_\infty$ when $p\rightarrow +\infty$.
A fundamental property of persistence diagrams is their stability with
respect to (small) perturbations of their originating functions.
Indeed, the {\em stability theorem}~\cite{Bauer13b,Chazal09a,Chazal16a,Cohen07}
asserts that 
for any
$f,g:X\to\R$, we have
%
\begin{equation}\label{eq:PD_stab}
d_\infty(\Dg(f),\, \Dg(g))\leq \|f-g\|_\infty,
\end{equation}

In practice, persistence diagrams can be used as descriptors for data
via the choice of appropriate filtering functions~$f$, e.g.  distance
to the data in the ambient space, eccentricity, curvature, etc. The
main strengths of the obtained descriptors are: (a) to be provably stable as
 mentioned previously; (b) to be invariant under reparametrization of the
data; and (c) to encode information about the topology of the data,
which is complementary and of an essentially different nature compared
to geometric or statistical quantities.
These properties have made persistence diagrams useful in a variety of
contexts, including the ones mentioned in the introduction of the
paper.
For further details on persistent homology and on applications of
persistence diagrams, the interested reader can refer
e.g. to~\cite{Oudot15} and the references therein.

\paragraph{Notation.} Let $\SpD$ be the space of persistence diagrams with at most countably many points,
$\SpfbD$ be the space of finite and bounded persistence diagrams, and
$\SpND$ be the space of bounded persistence diagrams with less than $N$ points.
Obviously, we have the following sequence of (strict) inclusions: $\SpND\subset\SpfbD\subset\SpD$.

\subsection{Kernel Methods}
\label{sec:kernelMethods}

\textbf{Positive Definite Kernels.} Given a set $X$, a function $k:X\times X\to\R$ 
is called a {\em positive definite kernel} if for all integers $n$, for all families $x_1,\cdots,x_n$ of points in $X$, 
the matrix $[k(x_i,x_j)]_{i,j}$ is itself positive semi-definite. For brevity we will refer to positive definite 
kernels as kernels in the rest of the paper.
It is known that kernels generalize scalar products, in the sense that, given a kernel $k$, there exists a Reproducing Kernel Hilbert Space
(RKHS) $\mathcal{H}_k$ and a {\em feature map} $\phi:X\to\mathcal{H}_k$ such that $k(x_1,x_2)=\langle \phi(x_1),\phi(x_2)\rangle_{\mathcal{H}_k}$.
A kernel $k$ also induces a distance $d_k$ on $X$ that can be computed as the Hilbert norm of the difference between two embeddings:
$$d_k^2(x_1,x_2)\eqdef k(x_1, x_1) + k(x_2, x_2) -2\,k(x_1, x_2).$$ 
We will be particularly interested in this distance, since one of the goals we will aim for will be that of designing a kernel 
$k$ for persistence diagrams such that $d_k$ has low distortion with respect to the first diagram distance $d_1$.

\paragraph*{Negative Definite and RBF Kernels.} A standard way to construct a kernel is to exponentiate the negative of a 
Euclidean distance. Indeed, the Gaussian kernel for vectors with parameter $\sigma>0$ does follow that template approach: 
$k_\sigma(x,y)={\rm exp}\left(-\frac{\|x-y\|^2}{2\sigma^2}\right)$. An important theorem of~\cite{Berg84} (Theorem 3.2.2, p.74) states that 
such an approach to build kernels, namely setting 
$$k_\sigma(x,y)\eqdef {\rm exp}\left(-\frac{f(x,y)}{2\sigma^2}\right),$$
for an arbitrary function $f$ can only yield a valid positive definite kernel for all $\sigma>0$ 
if and only if $f$ is a \emph{conditionally negative definite} function, namely that, for all integers $n$, 
 for all $x_1,\cdots,x_n \in X$, and for all $ a_1,\cdots,a_n \in \R$ such that $\sum_ia_i=0$, 
one has $\sum_{i,j}a_ia_jf(x_i,x_j)\leq 0$.

Unfortunately, as observed experimentally in Appendix~A of~\cite{Reininghaus14}, $d_1$  is not conditionally negative definite (in practice, it only suffices to sample a family of point 
clouds to observe experimentally that more often than not the inequality above will be violated for a particular weight vector $a$).
Actually, as observed in~\cite{Padellini17}, even the square of the diagram distances $d_p$ cannot be used to define Gaussian kernels.
Indeed, it was noted in Theorem~6 of~\cite{Feragen15} that, if the square of a distance $d$ defined on a geodesic space $X$ is conditionally negative definite,
then the metric space $X$ is flat, or CAT$(0)$. However, since the metric space $\SpD$, equipped with $d_p$, $p\in\N\cup\{+\infty\}$, is not CAT$(k)$ for any
$k>0$---which is due to the non-uniqueness of geodesics, see~\cite{Turner14}---it follows that $d_p^2$ is not conditionally negative definite.
 
In this article, we use an approximation of $d_1$ with the {\em Sliced Wasserstein distance}, which is provably 
conditionally negative definite, and we use it to define a RBF kernel that can be easily tuned thanks to its bandwidth parameter $\sigma$.

\subsection{Wasserstein distance for unnormalized measures on $\mathbb{R}$}
\label{sec:wasserstein}
The Wasserstein distance~\cite[\S6]{Villani09} is a distance between probability measures. 
For reasons that will become clear in the next section, we will focus on a variant of that distance: 
the 1-Wasserstein distance for \emph{nonnegative}, not necessarily normalized, measures on the real line~\cite[\S2]{Santambrogio15}. 
Let $\mu$ and $\nu$ be two nonnegative measures on the real line such that 
$|\mu|=\mu(\mathbb{R})$ and $|\nu|=\nu(\mathbb{R})$ are equal to the same number $r$. 
We define the three following objects:

\begin{align}
&\mathcal{W}(\mu,\nu)=\inf_{P\in\Pi(\mu,\nu)} \iint_{\mathbb{R}\times\mathbb{R}} |x-y| P({\rm d}x,{\rm d}y)\label{eq:optimal}\\
&\mathcal{Q}_r(\mu,\nu)=r \int_{\mathbb{R}} | M^{-1}(x)- N^{-1}(x)| {\rm d}x\label{eq:quantile}\\
&\mathcal{L}(\mu,\nu)=\inf_{f\in 1-\text{Lipschitz}}\int_{\mathbb{R}} f(x) [\mu({\rm d}x)-\nu({\rm d}x)]\label{eq:Kanto}
\end{align}
where $\Pi(\mu,\nu)$ is the set of measures on $\mathbb{R}^2$ with marginals $\mu$ and $\nu$, 
and $M^{-1}$ and $N^{-1}$ the generalized quantile functions of the probability measures $\mu/r$ and $\nu/r$ respectively. 

\begin{prop}
We have $\mathcal{W}=\mathcal{Q}_r=\mathcal{L}$. 
Additionally 
\emph{(i)} $\mathcal{Q}_r$ is conditionally negative definite on the space of measures of mass $r$; 
\emph{(ii)} for any three positive measures $\mu,\nu,\gamma$ such that $|\mu|=|\nu|$, 
we have $\mathcal{L}(\mu+\gamma,\nu+\gamma)=\mathcal{L}(\mu,\nu)$.
\end{prop}

\begin{proof} 
The equality between~(\ref{eq:optimal}) and~(\ref{eq:quantile}) is known for probability measures on the 
real line---see Proposition 2.17 in~\cite{Santambrogio15} for instance, and can be trivially generalized to unnormalized measures. 
The equality between~(\ref{eq:optimal}) and~(\ref{eq:Kanto}) is due to the well known Kantorovich duality for a distance 
cost~\cite[Particular case 5.4]{Villani09} which can also be trivially generalized to unnormalized measures, 
which proves the main statement of the proposition. 

The definition of $Q_r$ shows that the Wasserstein distance 
is the $l_1$ norm of $r M^{-1}- r N^{-1}$, and is therefore conditionally negative definite (as the $l_1$ distance 
between two direct representations of $\mu$ and $\nu$ as functions $r M^{-1}$ and $r N^{-1}$), proving point (i). 
The second statement is immediate.
\end{proof}


\begin{rem} For two unnormalized uniform empirical measures 
$\mu=\sum_{i=1}^n \delta_{x_i}$ and $\nu=\sum_{i=1}^n \delta_{y_i}$ of the same size, with ordered 
$x_1\leq \cdots \leq x_{n}$ and $y_1\leq \cdots \leq y_{n}$, one has:
$\mathcal{W}(\mu,\nu)=\sum_{i=1}^n|x_i-y_i|=\|X-Y\|_1$, where $X=(x_1,\cdots,x_n)\in\R^n$ and $Y=(y_1,\cdots,y_n)\in\R^n$.
\end{rem}



\section{The Sliced Wasserstein Kernel}

\subsection{The Sliced Wasserstein Kernel}

In this section we define a new kernel between persistence diagrams, called the {\em Sliced Wasserstein} kernel,
based on the Sliced Wasserstein metric of~\cite{Rabin11}. The idea underlying this
metric is to slice the plane with lines passing through the origin, to
project the measures onto these lines where $\mathcal W$ is computed, 
and to integrate those distances over all possible lines.  Formally:

\begin{defin}  
Given $\theta\in\R^2$ with $\|\theta\|_2=1$, let $L(\theta)$ denote the line $\{\lambda\,\theta : \lambda\in\R\}$, and
let $\pi_\theta:\R^2\rightarrow L(\theta)$ be the orthogonal projection onto $L(\theta)$.
Let $\Dg_1,\Dg_2$ be two persistence diagrams, and let $\mu_1^\theta=\sum_{p\in \Dg_1}\delta_{\pi_\theta(p)}$ and 
$\mu_{1\Delta}^\theta=\sum_{p\in \Dg_1}\delta_{\pi_\theta\circ\pi_\Delta(p)}$,
and similarly for $\mu_2^\theta$, where $\pi_\Delta$ is the orthogonal projection onto the diagonal.
Then, the {\em Sliced Wasserstein distance} 
is defined as:
$$\SW(\Dg_1,\Dg_2)\eqdef\frac{1}{2\pi}\int_{\mathbb{S}_1} \mathcal W(\mu_1^\theta+\mu_{2\Delta}^\theta,\mu_2^\theta+\mu_{1\Delta}^\theta){\rm d}\theta.$$
\end{defin}

Note that, by symmetry, one can restrict on the half-circle $[-\frac\pi 2,\frac\pi 2]$ and normalize by $\pi$ instead of $2\pi$.
Since $\mathcal Q_r$ is conditionally negative definite,
we can deduce that $\SW$ itself is conditionally negative definite:

\begin{lem}\label{lem:nd}
$\SW$ is conditionally negative definite on $\SpfbD$.
\end{lem}


\begin{proof}

Let $n\in\mathbb{N}^*$, $a_1,\cdots,a_n\in\R$ such that $\sum_ia_i=0$ and $\Dg_1,\cdots,\Dg_n\in \SpfbD$.
Given $1\leq i\leq n$, we let 
$\tilde\mu_i^\theta=\mu_i^\theta + \sum_{q\in \Dg_k,k\neq i}\delta_{\pi_\theta\circ\pi_\Delta(q)}$, 
$\tilde\mu_{ij\Delta}^\theta=\sum_{p\in \Dg_k,k\neq i,j}\delta_{\pi_\theta\circ\pi_\Delta(p)}$ and $d=\sum_i |\Dg_i|$.
Then:
\begin{align}
&\sum_{i,j} a_ia_j \mathcal W(\mu_i^\theta+\mu_{j\Delta}^\theta,\mu_j^\theta+\mu_{i\Delta}^\theta)
=\sum_{i,j} a_ia_j\mathcal L(\mu_i^\theta+\mu_{j\Delta}^\theta,\mu_j^\theta+\mu_{i\Delta}^\theta)\nonumber \\
&=\sum_{i,j} a_ia_j\mathcal L(\mu_i^\theta+\mu_{j\Delta}^\theta+\mu_{ij\Delta}^\theta,\mu_j^\theta+\mu_{i\Delta}^\theta+\mu_{ij\Delta}^\theta)\nonumber\\
&= \sum_{i,j} a_ia_j\mathcal L(\tilde\mu_i^\theta,\tilde\mu_j^\theta)
= \sum_{i,j} a_ia_j\mathcal Q_d(\tilde\mu_i^\theta,\tilde\mu_j^\theta)\leq 0\nonumber
\end{align}  
The result follows by linearity of integration.

\end{proof}

Hence, the theorem of~\cite{Berg84}
allows us to define a valid kernel with: 
\begin{equation}\label{eq:kSW}
\kSW(\Dg_1,\Dg_2)\eqdef{\rm exp}\left(-\frac{\SW(\Dg_1,\Dg_2)}{2\sigma^2}\right).
\end{equation}

\subsection{Metric Equivalence}

We now give the main theoretical result of
this article, which states that $\SW$ is {\em strongly equivalent} to $d_1$.  
This has to be compared with~\cite{Reininghaus15} and~\cite{Kusano16}, which
only prove stability and injectivity. Our equivalence result
states that $\kSW$, in addition to be stable
and injective, preserves the metric between persistence diagrams, which should
intuitively lead to an improvement of the classification power. This
intuition is illustrated in Section~\ref{sec:expe} and
Figure~\ref{fig:Airplanedistances}, where we show an improvement of
classification accuracies on several benchmark applications.

\subsubsection{Stability}

\begin{thm}\label{th:stab}
$\SW$  is stable with respect to $d_1$ on $\SpfbD$.
For any $\Dg_1,\Dg_2\in \SpfbD$, one has: $$\SW(\Dg_1,\Dg_2)\leq 2\sqrt{2}d_1(\Dg_1,\Dg_2).$$
\end{thm}

\begin{proof}
Let $\theta\in\R^2$ be such that $\|\theta\|_2=1$. Let $\Dg_1,\Dg_2\in \SpfbD$, and
let $\Dg_1^\theta = \{\pi_\theta(p):p\in \Dg_1\}\cup\{\pi_\theta\circ\pi_\Delta(q):q\in \Dg_2\}$ and 
$\Dg_2^\theta=\{\pi_\theta(q):q\in \Dg_2\}\cup\{\pi_\theta\circ\pi_\Delta(p):p\in \Dg_1\}$. 
Let $\gamma^*$ be the one-to-one bijection between $\Dg_1^\theta$ and $\Dg_2^\theta$
induced by $\mathcal W(\mu_1^\theta+\mu_{2\Delta}^\theta,\mu_2^\theta+\mu_{1\Delta}^\theta)$, and
let $\gamma$ be the 
one-to-one bijection between $\Dg_1\cup\pi_\Delta(\Dg_2)$ and $\Dg_2\cup\pi_\Delta(\Dg_1)$
induced by the partial bijection achieving $d_1(\Dg_1,\Dg_2)$.
Then $\gamma$ naturally induces a one-to-one matching $\gamma_\theta$
between $\Dg_1^\theta$ and $\Dg_2^\theta$ with:
$$\gamma_\theta=\{(\pi_\theta(p),\pi_\theta(q)):(p,q)\in\gamma\}\cup
\{(\pi_\theta\circ\pi_\Delta(p),\pi_\theta\circ\pi_\Delta(q)):(p,q)\in\gamma,\ p,q\not\in{\rm im}(\pi_\Delta)\}.$$

Now, one has the following inequalities:
\begin{align}
&\mathcal W(\mu_1^\theta+\mu_{2\Delta}^\theta,\mu_2^\theta+\mu_{1\Delta}^\theta) = \sum_{(x,y)\in\gamma^*} |x-y|\nonumber\\
&\leq \sum_{(\pi_\theta(p),\pi_\theta(q))\in\gamma_\theta} |\langle p,\theta\rangle-\langle q, \theta\rangle|
{\rm\ since\ }
\gamma_\theta{\rm\ is\ not\ the\ optimal\ matching\ between\ }\Dg_1^\theta{\rm\ and\ }\Dg_2^\theta\nonumber\\
&\leq \sum_{(\pi_\theta(p),\pi_\theta(q))\in\gamma_\theta} \|p-q\|_2\text{ by the Cauchy-Schwarz inequality since }\|\theta\|_2=1\nonumber\\
&\leq \sqrt{2}\sum_{(\pi_\theta(p),\pi_\theta(q))\in\gamma_\theta} \|p-q\|_\infty{\rm\ since\ }\|\cdot\|_2\leq\sqrt{2}\|\cdot\|_\infty\nonumber\\
&\leq 2\sqrt{2}\sum_{(p,q)\in\gamma} \|p-q\|_\infty{\rm\ since\ }\|\pi_\Delta(p)-\pi_\Delta(q)\|_\infty \leq \|p-q\|_\infty\nonumber\\
&= 2\sqrt{2}d_1(\Dg_1,\Dg_2)\nonumber
\end{align}

Hence, we have
$\SW(\Dg_1,\Dg_2)\leq 2\sqrt{2}d_1(\Dg_1,\Dg_2)$.
\end{proof}

We now prove the discriminativity of $\SW$.
For this, we need a stronger assumption on the persistence diagrams, namely their cardinalities have not only to be finite, but also bounded
by some $N\in\mathbb{N}^*$.

\subsubsection{Discriminativity}

\begin{thm}\label{th:discr}
$\SW$  is {\em discriminative} with respect to $d_1$ on $\SpND$.
For any $\Dg_1,\Dg_2\in X$, one has: $$\frac{1}{2M}d_1(\Dg_1,\Dg_2)\leq \SW(\Dg_1,\Dg_2),$$
where $M=1+2N(2N-1)$. 
\end{thm}

\begin{proof}
Let $\Dg_1,\Dg_2\in \SpND$. 
Let $\mathbb{S}^+_1\subseteq\mathbb{S}_1$ be the subset of the circle delimited by the angles $\left[-\frac{\pi}{2},\frac{\pi}{2}\right]$.
Let us consider the following set:
$$\Theta_1 = \left\{\theta\in \mathbb{S}^+_1:\exists p_1,p_2\in \Dg_1\text{ such that }\langle\theta, p_2-p_1\rangle=0\right\},$$
and similarly:
$$\Theta_2 = \left\{\theta\in \mathbb{S}^+_1:\exists q_1,q_2\in \Dg_2\text{ such that }\langle\theta, q_2-q_1\rangle=0\right\}.$$
Now, we let $\Theta=\Theta_1\cup\Theta_2\cup\left\{-\frac{\pi}{2},\frac{\pi}{2}\right\}$ be the union of these sets, 
and sort $\Theta$ in decreasing order.
One has $|\Theta|\leq 2N(2N-1)+2=M+1$ since a vector $\theta$ that is orthogonal to a line defined by a specific pair of 
points $(p_1,p_2)$ appears exactly once in $\mathbb{S}_1^+$.

For any $\theta$ that is between two consecutive $\theta_k,\theta_{k+1}\in\Theta$, the order of the projections 
onto $L(\theta)$ of the points of both $\Dg_1$ and $\Dg_2$ remains the same. Given any point $p\in \Dg_1\cup\pi_\Delta(\Dg_2)$, 
we let $\gamma(p)\in \Dg_2\cup\pi_\Delta(\Dg_1)$ be its matching point according 
to the matching given by $\mathcal W(\mu_1^\theta+\mu_{2\Delta}^\theta,\mu_2^\theta+\mu_{1\Delta}^\theta)$.
Then, one has the following equalities:

\begin{align}
\int_{\theta_k}^{\theta_{k+1}}&\mathcal W(\mu_1^\theta+\mu_{2\Delta}^\theta,\mu_2^\theta+\mu_{1\Delta}^\theta)\ {\rm d}\theta\nonumber\\
&=\int_{\theta_k}^{\theta_{k+1}}\underset{p\in \Dg_1\cup\pi_\Delta(\Dg_2)}{\sum}|\langle p-\gamma(p),\theta\rangle|\ {\rm d}\theta\nonumber\\
&=\underset{p\in \Dg_1\cup\pi_\Delta(\Dg_2)}{\sum}\|p-\gamma(p)\|_2\int_0^{\theta_{k+1}-\theta_k}|{\rm cos}\left(\alpha_p+\beta\right)
|\ {\rm d}\beta{\rm\ where\ }\alpha_p=\angle(p-\gamma(p),\theta_k)\nonumber
\end{align}

\begin{figure}\begin{center} 
\includegraphics[width=15cm]{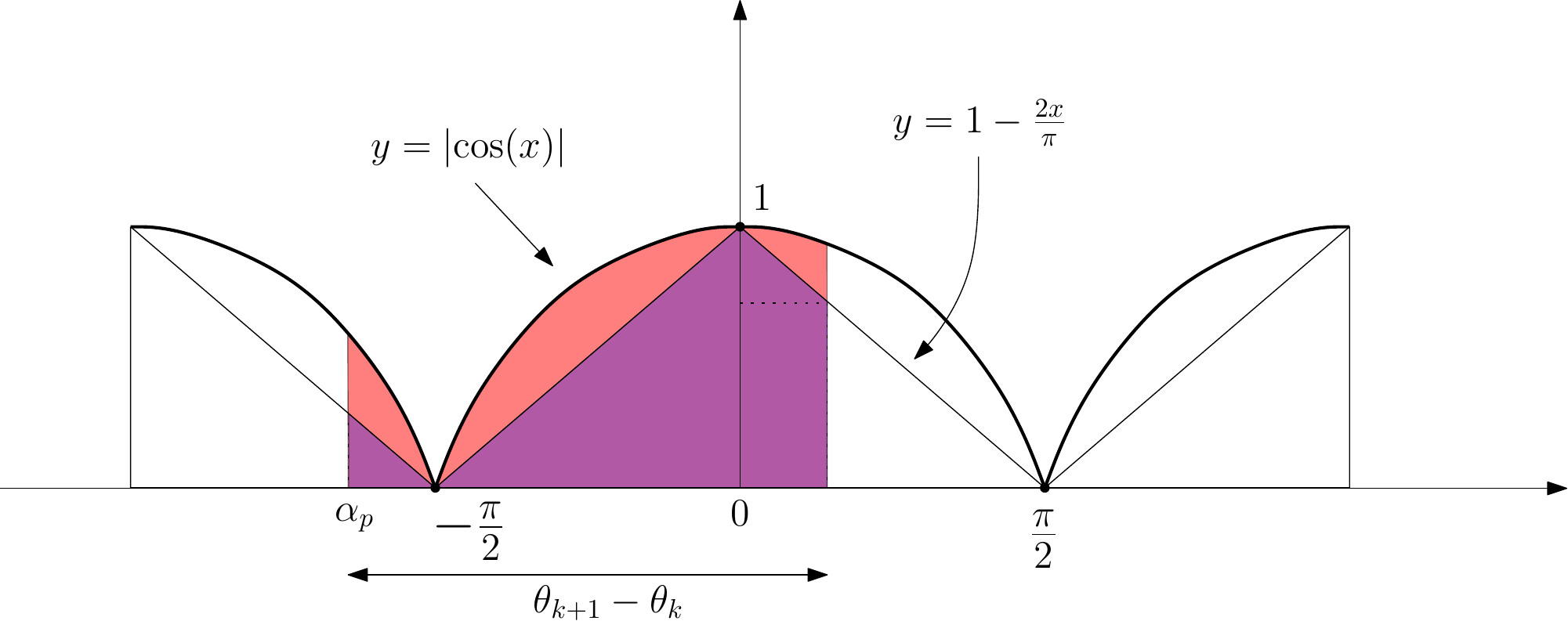}
\caption{\label{fig:cosineConc}
The integral of $|{\rm cos}(\cdot)|$ has a lower bound that depends on the length of the integral support.
In particular, when $\theta_{k+1}-\theta_k\leq\pi$, this integral is more than $\frac{\left(\theta_{k+1}-\theta_k\right)^2}{2\pi}$ 
by the Cauchy-Schwarz inequality.}
\end{center}\end{figure}

We need to lower bound $\int_0^{\theta_{k+1}-\theta_k}|{\rm cos}\left(\alpha_p+\beta\right)|d\beta$.
Since $\theta_{k+1}-\theta_k\leq\pi$, one can show that this integral cannot be less than $\frac{\left(\theta_{k+1}-\theta_k\right)^2}{2\pi}$ 
using cosine concavity---see Figure~\ref{fig:cosineConc}. 
Hence, we now have the following lower bound:
 
\begin{align}
\int_{\theta_k}^{\theta_{k+1}} &\mathcal W(\mu_1^\theta+\mu_{2\Delta}^\theta,\mu_2^\theta+\mu_{1\Delta}^\theta)\ {\rm d}\theta
\geq \frac{\left(\theta_{k+1}-\theta_k\right)^2}{2\pi}\underset{p\in \Dg_1\cup\pi_\Delta(\Dg_2)}{\sum}\|p-\gamma(p)\|_2\nonumber\\
&\geq \frac{\left(\theta_{k+1}-\theta_k\right)^2}{2\pi}\underset{p\in \Dg_1\cup\pi_\Delta(\Dg_2)}{\sum}\|p-\gamma(p)\|_\infty\geq\
\frac{\left(\theta_{k+1}-\theta_k\right)^2}{2\pi}\underset{\substack{ p\notin \pi_\Delta(\Dg_2) \\ {\rm\ or\ }\gamma(p)\notin\pi_\Delta(\Dg_1)} }{\sum}\|p-\gamma(p)\|_\infty\nonumber\\
&\geq \frac{\left(\theta_{k+1}-\theta_k\right)^2}{2\pi}d_1(\Dg_1,\Dg_2).
\nonumber
\end{align}

Let $\Theta=\left\{\theta_1=-\frac{\pi}{2},\theta_2,...,\theta_{|\Theta|}=\frac{\pi}{2}\right\}$. Then, one has:

\begin{align*}
\SW(\Dg_1,\Dg_2) & =\frac{1}{\pi}\int_{-\frac{\pi}{2}}^{\frac{\pi}{2}} \mathcal W(\mu_1^\theta+\mu_{2\Delta}^\theta,\mu_2^\theta+\mu_{1\Delta}^\theta)\ {\rm d}\theta
=\frac{1}{\pi}\sum_{k=1}^{|\Theta|-1}\int_{\theta_k}^{\theta_{k+1}} \mathcal W(\mu_1^\theta+\mu_{2\Delta}^\theta,\mu_2^\theta+\mu_{1\Delta}^\theta)\ 	{\rm d}\theta\nonumber\\
&\geq \left(\sum_{k=1}^{|\Theta|-1}\left(\theta_{k+1}-\theta_k\right)^2\right)\frac{d_1(\Dg_1,\Dg_2)}{2\pi^2} \nonumber\\
&\geq \frac{\pi^2}{|\Theta|-1}\frac{d_1(\Dg_1,\Dg_2)}{2\pi^2}\text{ by the Cauchy-Schwarz inequality} \nonumber\\
&\geq \frac{d_1(\Dg_1,\Dg_2)}{2M}
\end{align*}

Hence, $\SW$ is discriminative.
\end{proof}

In particular, Theorems~\ref{th:stab} and~\ref{th:discr} allow us to show that $d_{\rm SW}$, the distance induced by $k_{\rm SW}$ in its RKHS,
is also equivalent to $d_1$ in a broader sense: there exist continuous, positive and monotone functions $g,h$ such that $g(0)=h(0)=0$
and $h\circ d_1\leq d_{\rm SW}\leq g\circ d_1$. \\
 
The condition on the cardinalities of persistence diagrams can be relaxed. Indeed, one can prove that the feature map $\phi_{\rm SW}$ induced by $\kSW$ 
is injective when the persistence diagrams are only assumed to be finite and bounded:

\begin{prop}\label{prop:inj}
The feature map $\phi_{\rm SW}$ is continuous and injective with respect to $d_1$ on $\SpfbD$.
\end{prop}	

\begin{proof}
Note that if the persistence diagrams have bounded cardinalities, Proposition~\ref{prop:inj} is an immediate consequence of Theorem~\ref{th:discr}.
One has that $\phi_{\rm SW}$ is continous since $d_{\rm SW}$ is stable (cf Theorem~\ref{th:stab}).
Now, let $\Dg_1,\Dg_2\in \SpfbD$. 
such that  $d_{\rm SW}(\Dg_1,\Dg_2)=\|\phi_{\rm SW}(\Dg_1)-\phi_{\rm SW}(\Dg_2)\|=0$. 
We necessarily have ${\rm SW}(\Dg_1,\Dg_2)=0$.
Assume that $d_1(\Dg_1,\Dg_2)>0$. 
Then, there must be a point $p$ in $\Dg_1$ that is not in $\Dg_2$.
The Sliced Wasserstein distance being $0$, there must be, for every $\theta\in\mathbb{S}_1$, a point $q_\theta$ in $\Dg_2$ 
that has the same projection onto $L(\theta)$ as $p$: $\pi_\theta(q_\theta)=\pi_\theta(p)$, i.e. 
$q_\theta\in(\pi_\theta(p),p)$, the line defined by the pair $\pi_\theta(p),p$. 
All these lines $(\pi_\theta(p),p)$ intersect at $p\neq q_\theta$.
Thus, $q_{\theta_1}\neq q_{\theta_2}$ for any $\theta_1\neq \theta_2$, hence $\Dg_2$
must include an infinite number of points,
which is impossible. Thus, $d_1(\Dg_1,\Dg_2)=0$ and $\phi_{\rm SW}$ is injective. 

\end{proof}

In particular, $\kSW$ can be turned into a universal kernel by considering ${\rm exp}(\kSW)$  (cf Theorem~1 in~\cite{Kwitt15}).
This can be useful in a variety of tasks, including tests on distributions of persistence diagrams.

\subsection{Computation}\label{sec:comput}

\paragraph*{Approximate computation.} In practice, we propose to approximate $\kSW$ in $O(N{\rm log}(N))$
time using Algorithm~\ref{alg:aksw}. This algorithm first samples $M$ directions
in the half-circle $\mathbb{S}^+_1$; it then computes, for
each sample $\theta_i$ and for each persistence diagram $\Dg$, the scalar
products between the points of $\Dg$ and $\theta_i$, to sort them next in a
vector $V_{\theta_i}(\Dg)$. Finally, the $\ell_1$-norm between the vectors 
is averaged over the sampled directions:
${\rm SW}_M(\Dg_1,\Dg_2)
=\frac 1M \sum_{i=1}^M \|V_{\theta_i}(\Dg_1)-V_{\theta_i}(\Dg_2)\|_1.$
Note that one can easily adapt the proof of Lemma~\ref{lem:nd} to show that ${\rm SW}_M$ 
is negative semi-definite
by using the linearity of the sum. Hence, this approximation remains a kernel.
If the two persistence diagrams have cardinalities bounded by $N$,
then the running time of this procedure is $O(MN{\rm log}(N))$. 
This approximation of $\kSW$  is useful since, as shown in
Section~\ref{sec:expe}, we have observed empirically that just a few directions are sufficient to get good classification accuracies.

\begin{algorithm}
\caption{Approximate computation of $\SW$}
\label{alg:aksw}
\begin{algorithmic}
\STATE {\bfseries Input:} $\Dg_1=\{p^1_1,\cdots,p^1_{N_1}\}$, $\Dg_2=\{p^2_1,\cdots,p^2_{N_2}\}, M$.
\STATE Add $\pi_\Delta(\Dg_1)$ to $\Dg_2$ and vice-versa.
\STATE Let $\SW=0$; $\theta=-\pi/2$; $s=\pi/M$;
\FOR{$i=1,\cdots,M$}
	\STATE Store the products $\langle p_k^1,\theta\rangle$ in an array $V_1$;
	\STATE Store the
 products $\langle p_k^2,\theta\rangle$ in an array $V_2$;
	\STATE Sort $V_1$ and $V_2$ in ascending order;
	\STATE $\SW=\SW+s \|V_1-V_2\|_1$;
	\STATE $\theta= \theta + s$;
\ENDFOR
\STATE {\bfseries Output:} $(1/\pi)\SW$;
\end{algorithmic}
\end{algorithm}  

\paragraph*{Exact computation.} A persistence diagram is said to be in {\em general position} if it has no triplet of aligned points.  
If the persistence diagrams have cardinalities bounded by $N$, then the exact kernel computation for persistence diagrams in general position can be done in $O(N^2{\rm log}(N))$ time with 
Algorithm~\ref{alg:ksw}. In practice, given $\Dg_1$ and $\Dg_2$, we slightly modify them with infinitesimally small random perturbations. The resulting persistence diagrams 
$\tilde{\Dg}_1$ and $\tilde{\Dg}_2$ are in general position and we can approximate
$\kSW(\Dg_1,\Dg_2)$ with $\kSW(\tilde{\Dg}_1,\tilde{\Dg}_2)$.	

\begin{algorithm}
\caption{Exact computation of $\SW$}\label{alg:ksw}
\KwIn{$\Dg_1=\{p^1_1,\cdots,p^1_{N_1}\}$ with $|\Dg_1|=N_1$, $\Dg_2=\{p^2_1,\cdots,p^2_{N_2}\}$ with $|\Dg_2|=N_2$}
Let $\Theta^1=[],\Theta^2=[],V_1=[],V_2=[]$, $B_1=[[]\ ...\ []]$, $B_2=[[]\ ...\ []]$, $\SW=0$;\\
\For{$i=1,\cdots,N_1$}{
    Add $p^2_{N_2+i}=\pi_\Delta(p^1_i)$ to $\Dg_2$;
  }
\For{$i=1,\cdots,N_2$}{
    Add $p^1_{N_1+i}=\pi_\Delta(p^2_i)$ to $\Dg_1$;
  }
\For{$i=1,2$}{
  \For{$j=1,\cdots,N_1+N_2-1$}{
    \For{$k=j+1,\cdots,N_1+N_2$}{
      Add $\angle \left[p^i_j-p^i_k\right]^\perp \in \left[-\frac{\pi}{2},\frac{\pi}{2}\right]$ to $\Theta^i$;
    }
  }
  Sort $A^i$ in ascending order;\\
  \For{$j=1,\cdots,N_1+N_2$}{
    Add $\langle p_j^i,[0,-1]\rangle$ to $V_i$;
  }
  Sort $V_i$ in ascending order;\\
  Let $f_i:p^i_j\mapsto{\rm position\ of\ }\left(p_j^i,-\frac{\pi}{2}\right){\rm\ in\ }V_i$; \\
  \For{$j=1,\cdots,(N_1+N_2)(N_1+N_2-1)/2$}{
    Let $k_1,k_2$ such that $\Theta^i[j]=\angle \left[p^i_{k_1}-p^i_{k_2}\right]^\perp$;\\
    Add $\left(p^i_{k_1},\Theta^i[j]\right)$ to $B_i\left[f_i(p^i_{k_1})\right]$; Add $\left(p^i_{k_2},\Theta^i[j]\right)$ to $B_i\left[f_i(p^i_{k_2})\right]$;\\
    Swap $f_i(p^i_{k_1})$ and $f_i(p^i_{k_2})$;
  }
  \For{$j=1,\cdots,N_1+N_2$}{
    Add $\left(p^i_j,\frac{\pi}{2}\right)$ to $B_i\left[f_i(p_j^i)\right];$
  }
}
\For{$i=1,\cdots,N_1+N_2$}{
  Let $k_1=0$, $k_2=0$;\\
  Let $\theta_m=-\frac{\pi}{2}$ and $\theta_M={\rm min}\{B_1[i][k_1]_2,B_2[i][k_2]_2\}$;\\
  \While{$\theta_m\neq \frac{\pi}{2}$}{
  $\SW = \SW+\|B_1[i][k_1]_1-B_2[i][k_2]_1\|_2\int_{0}^{\theta_M-\theta_m}{\rm cos}(\angle\left(B_1[i][k_1]_1-B_2[i][k_2]_1,\theta_m\right)+\theta){\rm d}\theta$;\\
  $\theta_m=\theta_M$;\\
  {\bf{\text if }} $\theta_M==B_1[i][k_1]_2$ {\bf {\text then }}$k_1=k_1+1$; {\bf{\text else }}$k_2=k_2+1$;\\
  $\theta_M={\rm min}\{B_1[i][k_1]_2,B_2[i][k_2]_2\}$;
  }
}
{\bf{\text return }}$\frac{1}{\pi}\SW$;
\end{algorithm}

\section{Experiments}
\label{sec:expe}

In this section, we compare $\kSW$ to $\kPSS$ and $\kPWG$ on 
several benchmark applications for which persistence diagrams have been proven useful. We compare these kernels in terms of classification 
accuracies and compuational cost. We review first our experimental setting, and review these tasks one by one.

\paragraph*{Experimental setting}
 All kernels are handled with the LIBSVM~\cite{Chang01} implementation of $C$-SVM, and results are averaged over 10 runs
on a 2.4GHz Intel Xeon E5530 Quad Core.
The cost factor $C$ is cross-validated in the following grid: $\{0.001, 0.01,0.1, 1,10,100,1000\}$.
Table~\ref{table:sum} summarizes the properties of the datasets we consider, namely number of labels, as well as training and test instances 
for each task. Figure~\ref{fig:taskltm} and \ref{fig:task2} illustrate how we use persistence diagrams to represent complex data.
We first describe the two baselines we considered, along with their parameterization, followed by our proposal.

\begin{table}[t]
\vskip 0.15in
\begin{center}
\begin{small}
\begin{sc}
\begin{tabular}{|l|c|c|c|}
\hline
 Task &        Training &                               Test &                       Labels \\
\hline
Orbit &        175 &                                    75 &                         5  \\
Texture &      240 &                                    240 &                        24  \\
Human &        415 &                                    1618 &                       8 \\
Airplane &     300 &                                    980 &                        4 \\
Ant &          364 &                                    1141 &                       5 \\
Bird &         257 &                                    832 &                        4 \\
FourLeg &      438 &                                    1097 &                       6 \\
Octopus &      334 &                                    1447 &                       2 \\
Fish &         304 &                                    905 &                        3 \\
\hline          
\end{tabular}
\end{sc}
\end{small}
\caption{\label{table:sum} Number of instances in the training set, the test set and number of labels.}
\end{center}
\vskip -0.1in
\end{table}

\begin{table}[t]
\vskip 0.15in
\begin{center}
\begin{small}
\begin{sc}

\begin{tabular}{|l|lll|}
\hline 
Task &         $\kPSS$ ($10^{-3}$)&    $\kPWG$ (1000) &                   $\kSW$ (6)                                  \\
\hline 
Orbit &        $63.6\pm1.2$ &          $77.7\pm1.2$ &                     ${\bf 83.7}\pm0.5$                           \\        
Texture &      ${\bf 98.8}\pm 0.0$ &   $95.8\pm0.0$ &                     $96.1\pm0.4$                          \\                                           
\hline 
Task &         $\kPSS$ &               $\kPWG$ &                          $\kSW$                           \\
\hline 
Human &        $68.5\pm2.0$ &          $64.2\pm1.2$ &                     ${\bf 74.0}\pm0.2 $ \\
Airplane &     $65.4\pm2.4$ &          $61.3\pm2.9$ &                     ${\bf 72.6}\pm0.2$  \\
Ant &          $86.3\pm1.0$ &          $87.4\pm0.5$ &                     ${\bf 92.3}\pm0.2$  \\
Bird &         $67.7\pm1.8$ &          ${\bf72.0}\pm1.2$ &                $67.0\pm0.5$  \\
FourLeg &      $67.0\pm2.5$ &          $64.0\pm0.6$ &                     ${\bf73.0}\pm0.4$ \\
Octopus &      $77.6\pm1.0$ &          $78.6\pm1.3$ &                     ${\bf85.2}\pm0.5$  \\
Fish &         $76.1\pm1.6$ &          ${\bf79.8}\pm0.5$ &                $75.0\pm0.4$ \\
\hline                                                            
\end{tabular}
\end{sc}
\end{small}

\caption{\label{table:Acc} Classification accuracies (\%) for the benchmark applications.}
\end{center}
\vskip -0.1in
\end{table}

\begin{table}[t]
\vskip 0.15in
\begin{center}
\begin{small}
\begin{sc}
\begin{tabular}{|l|llll|}
\hline 
Task &         $\kPSS$ ($10^{-3}$) &  $\kPWG$ (1000) &    $\kSW$ (6) &         \\
\hline 
Orbit &        $N(124\pm8.4)$ &       $N(144\pm14)$ &     $415\pm7.9+NC$ &                     \\        
Texture &      $N(165\pm27)$ &        $N(101\pm9.6)$ &    $482\pm68+NC$ &                      \\                                           
\hline 
Task &         $\kPSS$ &              $\kPWG$ &           $\kSW$ &           $\kSW$ (10) \\
\hline 
Human &        $N(29\pm0.3)$ &        $N(318\pm22)$ &     $2270\pm336+NC$ &  $107\pm14+NC$ \\
Airplane &     $N(0.8\pm0.03)$ &      $N(5.6\pm0.02)$ &   $44\pm5.4+NC$ &    $10\pm1.6+NC$ \\
Ant &          $N(1.7\pm0.01)$ &      $N(12\pm0.5)$ &     $92\pm2.8+NC$ &    $16\pm0.4+NC$ \\
Bird &         $N(0.5\pm0.01)$ &      $N(3.6\pm0.02)$ &   $27\pm1.6+NC$ &    $6.6\pm0.8+NC$ \\
FourLeg &      $N(10\pm0.07)$ &       $N(113\pm13)$ &     $604\pm25+NC$ &    $52\pm3.2+NC$ \\
Octopus &      $N(1.4\pm0.01)$ &      $N(11\pm0.8)$ &     $75\pm1.4+NC$ &    $14\pm2.1+NC$ \\
Fish &         $N(1.2\pm0.004)$ &     $N(9.6\pm0.03)$ &   $72\pm4.8+NC$ &    $12\pm1.1+NC$ \\
\hline                                                            
\end{tabular}

\end{sc}
\end{small}
\caption{\label{table:Gram} Gram matrices computation time (s) for the benchmark applications.
As explained in the text, $N$ represents the size of the set of possible parameters, and we have $N=13$ for $\kPSS$, 
$N=5\times5\times5=125$ for $\kPWG$ and $N=3\times5=15$ for $\kSW$. $C$ is a constant that depends only on the
training size. In all our applications, it is less than $0.1$s.}
\end{center}
\vskip -0.1in
\end{table}

\begin{figure*}[t] 
\centering
\includegraphics[width=0.98\textwidth]{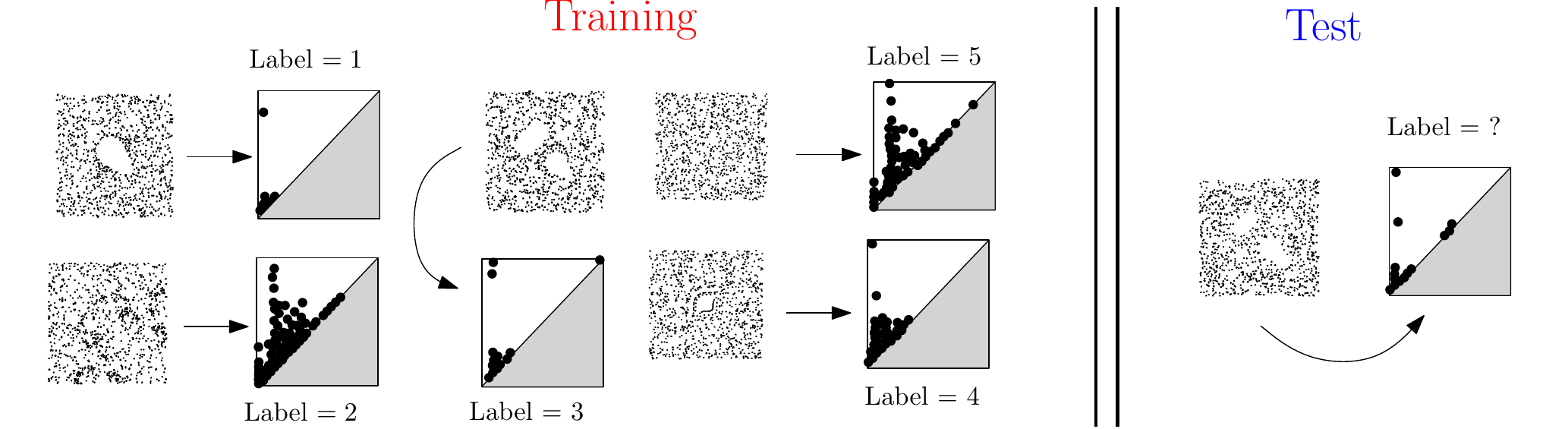}
\caption{\label{fig:taskltm} Sketch of the orbit recognition task. Each parameter $r$ in the 5 possible choices
leads to a specific behavior of the orbit. 
The goal is to recover parameters from the persistent homology of orbits in the test set.}
\end{figure*}

\begin{figure*}[t] 
\centering
\includegraphics[width=0.98\textwidth]{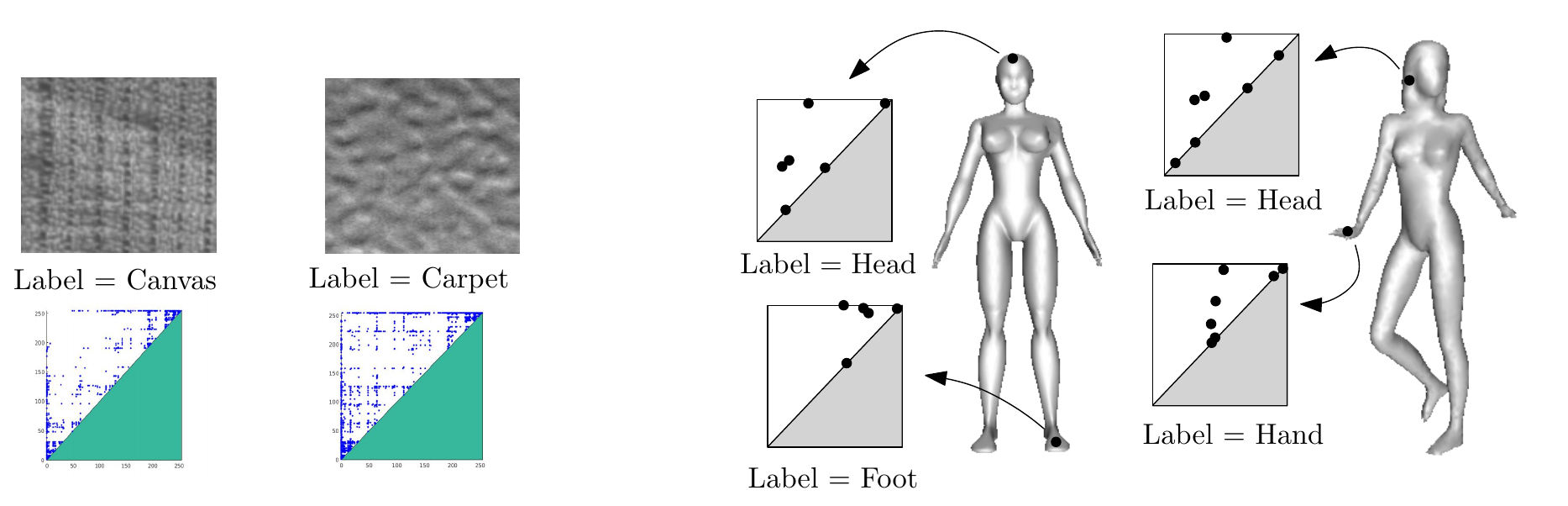}
\caption{\label{fig:task2} Examples of persistence diagrams computed on texture images from the \emph{OUTEX00000} dataset
and persistence diagrams computed from points on 3D shapes. One can see that corresponding points in different shapes have
similar persistence diagrams.}
\end{figure*}

\begin{figure}\centering
\includegraphics[width=7.5cm]{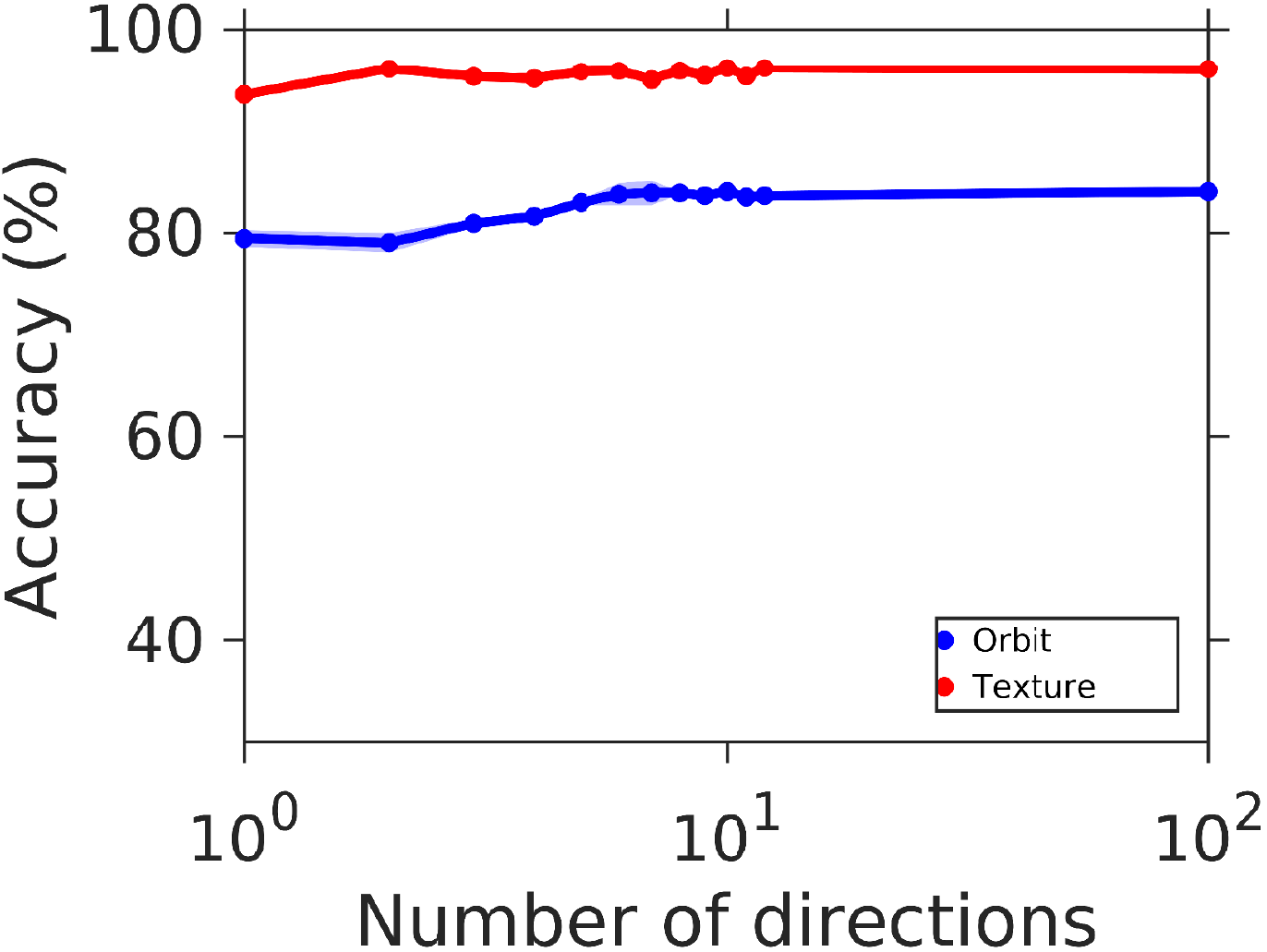}\ \ \includegraphics[width=7.5cm]{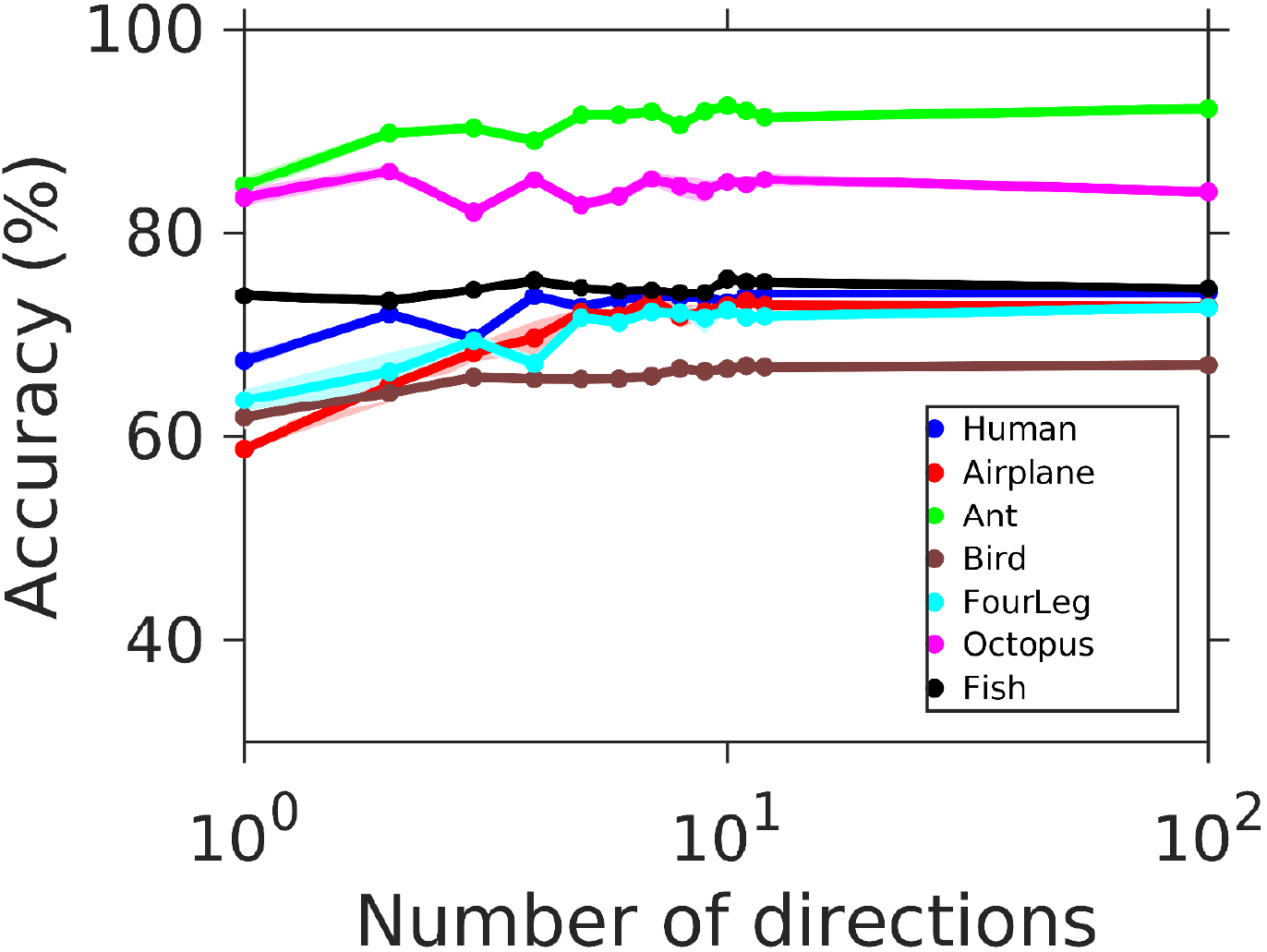} \\
\includegraphics[width=7.5cm]{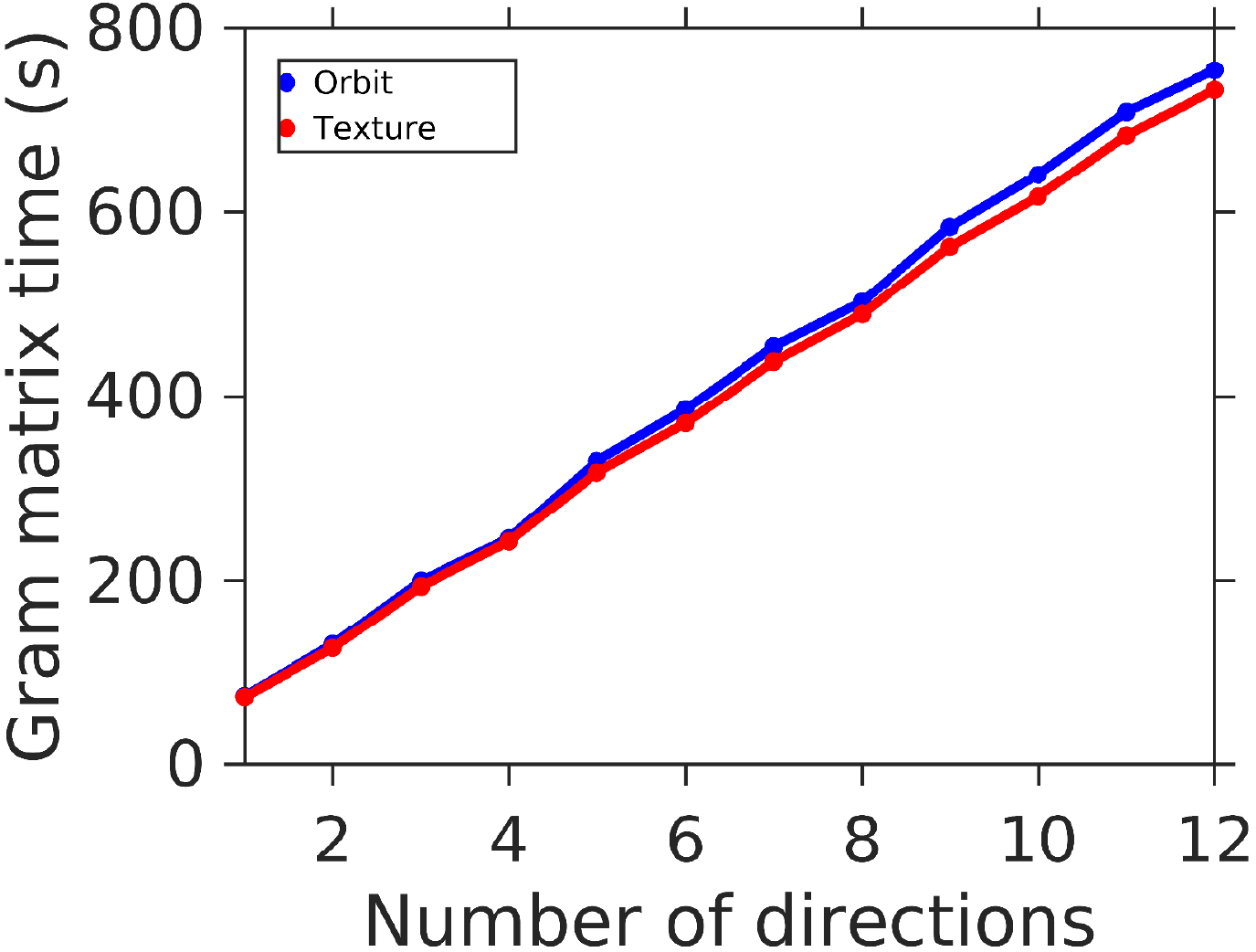} \ \ \includegraphics[width=7.5cm]{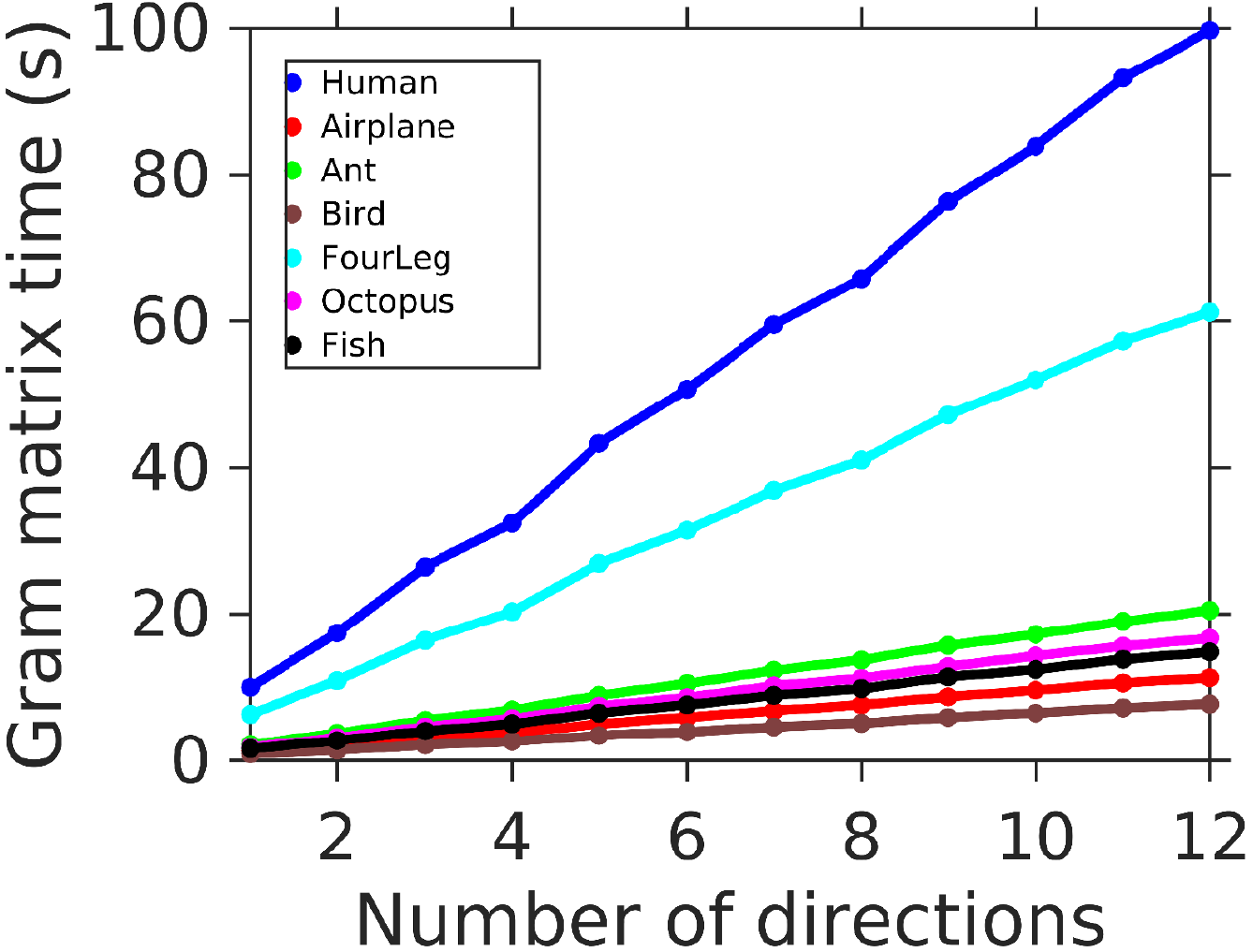} \\
\includegraphics[width=7.5cm]{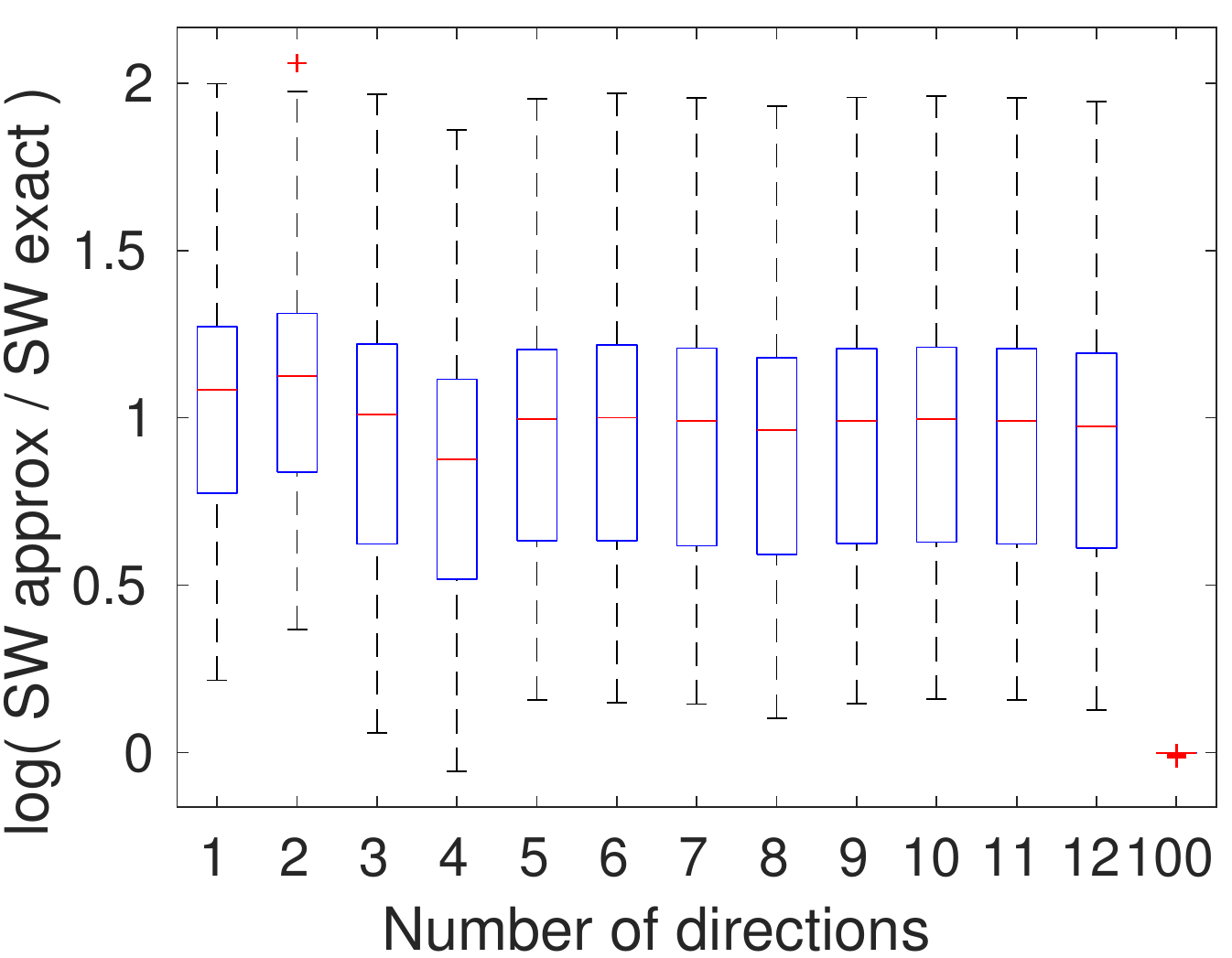}\ \ \includegraphics[width=7.5cm]{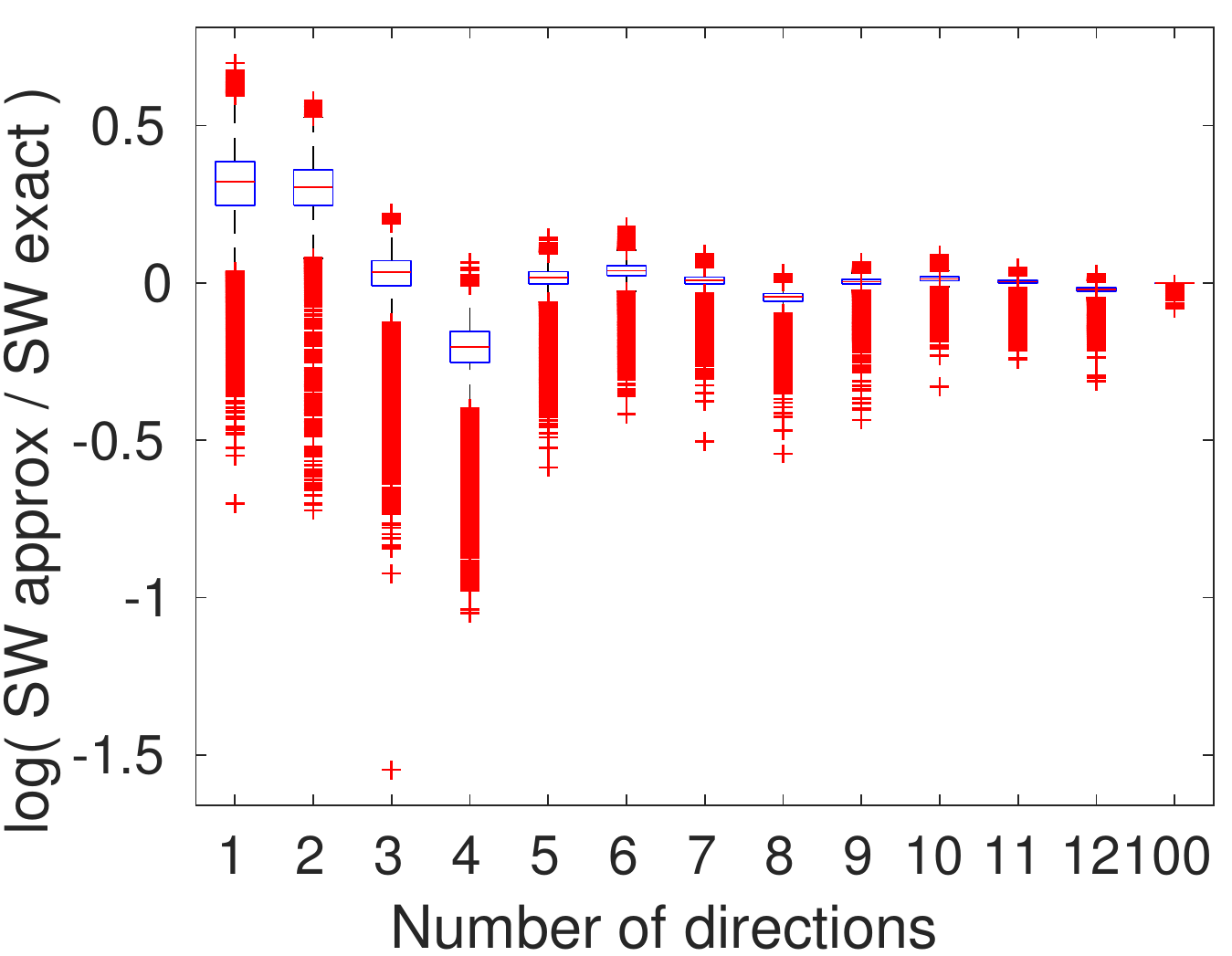}
\caption{\label{fig:plots} The first column corresponds to the orbit recognition and the texture classification while the second
column corresponds to 3D shape segmentation. 
On each column, the first row shows the dependence of the accuracy on the number of directions,
the second row shows the dependence of a single Gram matrix computation time, and the third row
shows the dependence of the ratio of the approximation of $\SW$ and the exact $\SW$.
Since the box plot of the ratio for orbit recognition is very similar to that of 3D shape segmentation,
we only give the box plot of texture classification in the first column. }
\end{figure}

\paragraph*{PSS.} The {\em Persistence Scale Space} kernel $\kPSS$~\cite{Reininghaus15} is
defined as the scalar product of the two solutions of the heat diffusion equation
with initial Dirac sources located at the points of the persistence diagram. It has the following closed form expression:
$$\kPSS(\Dg_1,\Dg_2)=\frac{1}{8\pi t}\sum_{p\in \Dg_1}\sum_{q\in \Dg_2} {\rm exp}\left(-\frac{\|p-q\|^2}{8t}\right) 
- {\rm exp}\left(-\frac{\|p-\bar q\|^2}{8t}\right),$$
where $\bar q=(y,x)$ is the symmetric of $q=(x,y)$ along the diagonal.
Since there is no clear heuristic on how to tune $t$, this parameter is chosen in the applications 
by ten-fold cross-validation with random 50\%-50\%
training-test splits and with the following set of $N_{\rm PSS}=13$ values: $0.001$, $0.005$, $0.01$, $0.05$, 
$0.1$, $0.5$, $1$, $5$, $10$, $50$, $100$, $500$ and $1000$. 

\paragraph*{PWG.} Let $K,p>0$ and $\Dg_1$ and $\Dg_2$ be two persistence diagrams.
Let $k_\rho$ be the Gaussian kernel with parameter $\rho >0$.
Let $\mathcal{H}_\rho$ be the RKHS associated to $k_\rho$.

Let $\mu_1=\sum_{x\in \Dg_1}{\rm arctan}(K{\rm pers}(x)^p)k_\rho(\cdot,x)\in\mathcal{H}_\rho$ 
be the kernel mean embedding of $\Dg_1$ weigthed by the diagonal distances.
Let $\mu_2$ be defined similarly. Let $\tau >0$.
The {\em Persistence Weighted Gaussian} kernel $\kPWG$~\cite{Kusano16, Kusano17} is
defined as the Gaussian kernel with parameter $\tau$ on $\mathcal{H}_\rho$:
$$\kPWG(\Dg_1,\Dg_2)={\rm exp}\left(-\frac{\|\mu_1-\mu_2\|_{\mathcal{H}_\rho}}{2\tau^2}\right).$$
The authors in~\cite{Kusano16} provide heuristics to compute $K$, $\rho$ and $\tau$
and give a rule of thumb to tune $p$. Hence, in the applications we select $p$ according to the rule of thumb, and
 we use ten-fold cross-validation with random 50\%-50\% training-test splits to chose $K$, $\rho$ and $\tau$. 
The ranges of possible values is obtained 
by multiplying the values computed with the heuristics
with the following range of $5$ factors: $0.01$, $0.1$, $1$, $10$ and $100$,
leading to $N_{\rm PWG}=5\times 5\times 5=125$ different sets of parameters. \\

\paragraph*{Parameters for $\kSW$.} The kernel we propose has only one parameter, the bandwidth $\sigma$ in Eq.~\ref{eq:kSW}, 
which we choose
using ten-fold cross-validation with random 50\%-50\% training-test splits. 
The range of possible values is obtained by computing the squareroot of the median, the first and the last deciles 
of all $\SW(\Dg_i,\Dg_j)$ in the training set,
then by multiplying these values by the following range of $5$ factors: $0.01$, $0.1$, $1$, $10$ and $100$, 
leading to $N_{\rm SW}=5\times 3= 15$ possible values.

\paragraph*{Parameter Tuning.} The bandwidth of $\kSW$ is, in practice, easier to tune than the parameters of its two competitors
when using grid search. Indeed, as is the case for all infinitely divisible kernels, the Gram matrix does not need to be 
recomputed for each choice of $\sigma$, since it only suffices to compute all the Sliced Wasserstein distances between persistence diagrams 
in the training set once. On the contrary, neither $\kPSS$ nor $\kPWG$ share this property, and require recomputations for each hyperparameter choice.
Note however that this improvement may no longer hold if one uses other methods to tune parameters.  
For instance, using $\kPWG$ without cross-validation is possible with the heuristics given by the authors in~\cite{Kusano16}, 
and leads to smaller training times, but also to worse accuracies.

\subsection{3D shape segmentation}
Our first task, whose goal is to produce point classifiers for 3D shapes, follows that presented in~\cite{Carriere15a}.

\paragraph*{Data.}  We use some categories of the mesh segmentation benchmark of Chen et al.~\cite{Chen09},
which contains 3D shapes classified in several categories (``airplane'', ``human'', ``ant''...).
For each category, our goal is to design a classifier that can assign, to each point in the shape,
a label that describes the relative location of that point in the shape. For instance, possible labels are, for the human category, 
``head'', ``torso'', ``arm''...
To train classifiers, we compute a persistence diagram per point using the geodesic distance function to this point---see~\cite{Carriere15a} for details.
We use 1-dimensional persistent homology (0-dimensional would not be informative since the shapes are connected,
leading to solely one point with coordinates $(0,+\infty)$ per persistence diagram). 
For each category, the training set contains one hundredth of the points of the first five 3D shapes,
and the test set contains one hundredth of the points of the remaining shapes in that category. Points in
training and test sets are evenly sampled. See Figure~\ref{fig:task2}.
Here, we focus on comparison between persistence diagrams, and not
on achieving state-of-the-art results. It has been proven that persistence diagrams bring complementary information
to classical descriptors in this task---see~\cite{Carriere15a}, 
hence reinforcing their discriminative power with appropriate kernels is of great interest.
Finally, since data points are in $\R^3$, we set the $p$ parameter of $\kPWG$ to $5$. 

\paragraph*{Results.} Classification accuracies are given in Table~\ref{table:Acc}.
For most categories, $\kSW$ outperforms competing kernels by a significant margin.
The variance of the results over the run is also less than that of its competitors. 
However, training times are not better in general. 
Hence, we also provide the results for an approximation of $\kSW$ with $10$ directions.
As one can see from Table~\ref{table:Acc} and from Figure~\ref{fig:plots}, this approximation leaves the accuracies almost unchanged, 
while the training times become comparable with the ones of the other competitors. Moreover, 
according to Figure~\ref{fig:plots}, using even less directions
would slightly decrease the accuracies, but still outperform the competitors performances,
while decreasing even more the training times.

\subsection{Orbit recognition}\label{sec:expeorbit}

In our second experiment, we use synthetized data.
The goal is to retrieve parameters of dynamical system orbits,
following an experiment proposed in~\cite{Adams17}.

\paragraph*{Data.} We study the {\em linked twist map}, a discrete dynamical system modeling
fluid flow. It was used in~\cite{Hertzsch07} to model flows in DNA microarrays.
Its orbits can be computed given a parameter $r>0$ and
initial positions $(x_0,y_0)\in[0,1]\times[0,1]$ as follows:

\[\left\{\begin{array}{l} x_{n+1} = x_n + ry_n(1-y_n)\ \ \ \ \ \ \ \ \ \ \ {\rm mod}\ 1 \\ y_{n+1} = y_n + rx_{n+1}(1-x_{n+1})\ \ \ {\rm mod}\ 1 \end{array}\right.\]

Depending on the values of $r$, the orbits may exhibit very different behaviors. For instance,
as one can see in Figure~\ref{fig:taskltm}, when $r$ is 3.5, there seems to be no interesting topological features
in the orbit, while voids form for $r$ parameters around 4.3.
Following~\cite{Adams17}, we use 5 different parameters $r=2.5,3.5,4,4.1,4.3$, that act as labels.
For each parameter, we generate 100 orbits with 1000 points and random initial positions. We then compute
the persistence diagrams of the distance functions to the point clouds with the GUDHI library~\cite{gudhi} and we use them (in all homological dimensions) 
to produce an orbit classifier
that predicts the parameter values, by training over a 70\%-30\% training-test split of the data.
Since data points are in $\R^2$, we set the $p$ parameter of $\kPWG$ to $4$. \\

\paragraph*{Results.} Since the persistence diagrams contain thousands of points, we use kernel approximations
to speed up the computation of the Gram matrices.
In order for the approximation error to be bounded by $10^{-3}$, 
we use an approximation of $\kSW$ with $6$ directions (as one can see from Figure~\ref{fig:plots}, 
this has a small impact on the accuracy), we approximate $\kPWG$ with $1000$ random Fourier features~\cite{Rahimi08},
and we approximate $\kPSS$ using Fast Gauss Transform~\cite{Morariu09}
with a normalized error of $10^{-10}$. 
One can see from Table~\ref{table:Acc} that the accuracy is increased a lot with $\kSW$.
Concerning training times, there is also a large improvement since we tune the parameters with grid search. 
Indeed, each Gram matrix needs not be recomputed for each parameter when using $\kSW$.

\subsection{Texture classification}

Our last experiment is inspired from~\cite{Reininghaus15} and~\cite{Li14}. 
We use the \emph{OUTEX00000} data base~\cite{Ojala02} for texture classification. 

\paragraph*{Data.} persistence diagrams are obtained for each texture image by computing first the sign component of CLBP descriptors~\cite{Guo10} 
with radius $R=1$ and $P=8$ neighbors for each image,
and then compute the persistent homology of this descriptor using the GUDHI library~\cite{gudhi}. 
See Figure~\ref{fig:task2}.
Note that, contrary to the experiment of~\cite{Reininghaus15}, we do not downsample the images to $32\times 32$ images,
but keep the original $128\times128$ images. 
Following~\cite{Reininghaus15}, we restrict the focus to 0-dimensional persistent homology.
We also use the first 50\%-50\% training-test split given in the database to produce classifiers. 
Since data points are in $\R^2$, we set the $p$ parameter of $k_{\rm PWG}$ to $4$. 

\paragraph*{Results} We use the same approximation procedure as in Section~\ref{sec:expeorbit}.
According to Figure~\ref{fig:plots}, even though the approximation of ${\rm SW}$ is rough,
this has again a small impact on the accuracy, while reducing the training time by a significant margin.
As one can see from Table~\ref{table:Acc}, 
using $\kPSS$ leads to almost state-of-the-art results~\cite{Ojala02, Guo10},
closely followed by the accuracies of $\kSW$ and $\kPWG$.
The best timing is given by $\kSW$, again because we use grid search. 
Hence, $\kSW$ almost achieves the best result, and its training time is
better than the ones of its competitors, due to the grid search parameter tuning.

\begin{figure}\centering
\includegraphics[width=8cm]{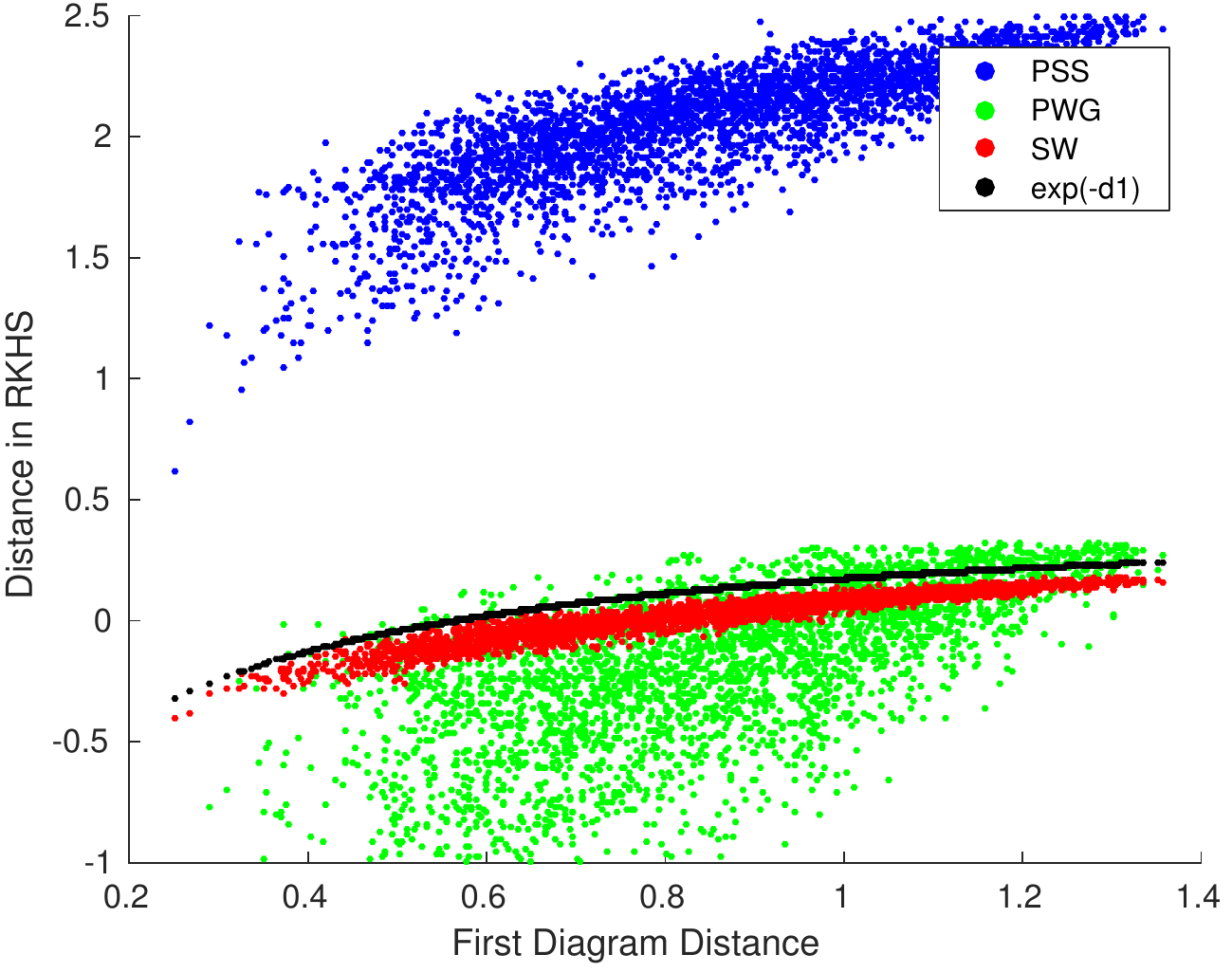}
\caption{\label{fig:Airplanedistances} We show how the metric $d_1$ is distorted.
Each point represents a pair of persistence diagrams and its abscissae is the first diagram distance between them. 
Depending on the point color, its ordinate is the logarithm of the distance between persistence diagrams in the RKHS induced by either
$\kPSS$ (blue points), $\kPWG$ (green points), $\kSW$ (red points) and a Gaussian kernel on $d_1$ (black points).  }
\end{figure}

\subsection{Metric Distortion.} 

To illustrate the equivalence theorem, we also show in Figure~\ref{fig:Airplanedistances} 
a scatter plot where each point represents the comparison of two persistence diagrams taken from the Airplane segmentation data set. 
Similar plots can be obtained with the other datasets considered here.
For all points, the x-axis quantifies the first diagram distance $d_1$ for that pair,
while the y-axis is the logarithm of the RKHS distance induced by either $\kSW$, $\kPSS$, $\kPWG$ 
or a Gaussian kernel directly applied to $d_1$, to obtain comparable quantities. 
We use the parameters given by the cross-validation procedure described above.
One can see that the distances induced by $\kSW$ are less spread than the others,
suggesting that the metric induced by $\kSW$ is more discriminative.
Moreover the distances given by $\kSW$ and the Gaussian kernel on $d_1$ exhibit the same behavior, 
suggesting that $\kSW$ is the best natural equivalent of a Gaussian kernel for persistence diagrams.

\section{Conclusion}

In this article, we introduce the {\em Sliced Wasserstein kernel},
a new kernel for persistence diagrams that is provably {\em equivalent} to the first
diagram distance between persistence diagrams. We provide fast algorithms to approximate it,
and show on several datasets substantial improvements in accuracy and training times 
(when tuning parameters is done with grid search) over competing kernels. 
A particularly appealing property of that kernel is that it is infinitely divisible, 
substantially facilitating the tuning of parameters through cross validation.

\paragraph{Acknowledgements.}
SO was supported by ERC grant Gudhi and by ANR project TopData. 
MC was supported by a {\em chaire de l'IDEX Paris Saclay}.

\bibliography{biblio}
\bibliographystyle{plain}

\end{document}